\documentclass[version=preprint]{iacrcc}

\usepackage{lmodern}
\usepackage{microtype}
\usepackage[T1]{fontenc}
\usepackage{algorithm}
\usepackage{algpseudocode}
\usepackage{booktabs}
\usepackage[inline]{enumitem}
\usepackage{amsmath,amssymb}
\usepackage{tikz}
\usepackage{pgfplots}
\pgfplotsset{compat=1.18}
\usetikzlibrary{arrows.meta,positioning,fit,backgrounds}

\newcommand{\Z}{\mathbb{Z}}
\newcommand{\Rq}{R_q}
\newcommand{\bA}{\mathbf{A}}
\newcommand{\by}{\mathbf{y}}
\newcommand{\bz}{\mathbf{z}}
\newcommand{\bs}{\mathbf{s}}
\newcommand{\bw}{\mathbf{w}}
\newcommand{\HighBits}{\mathsf{HighBits}}
\newcommand{\LowBits}{\mathsf{LowBits}}
\newcommand{\wone}{w_1}
\newcommand{\rzero}{r_0}
\newcommand{\BCC}{\mathsf{BCC}}

\newcommand{\Decompose}{\mathsf{Decompose}}
\newcommand{\MakeHint}{\mathsf{MakeHint}}
\newcommand{\UseHint}{\mathsf{UseHint}}
\newcommand{\Hash}{\mathsf{H}}
\newcommand{\infnorm}[1]{\lVert #1\rVert_\infty}
\newcommand{\getsr}{\leftarrow\!\!\text{\tiny\$}\,}
\newcommand{\sthr}[1]{[\![#1]\!]}
\newcommand{\IH}{\mathsf{IH}}
\newcommand{\bh}{\mathbf{h}}
\newcommand{\bt}{\mathbf{t}}
\newcommand{\bu}{\mathbf{u}}
\newcommand{\br}{\mathbf{r}}

\newcommand{\pk}{\mathsf{pk}}

\title[running={TALUS: Threshold ML-DSA}]{TALUS: FIPS-204-Exact Threshold ML-DSA
  via Boundary Clearance}

\addauthor[inst={1},
           email={leo@codebat.ai},
          ]{Leo Kao}
\addauthor[inst={1}]{Raymond Chang}
\addaffiliation{Codebat Technologies Inc.}

\keywords{threshold signatures, ML-DSA, FIPS 204, post-quantum cryptography,
          lattice-based cryptography, boundary clearance, carry elimination}

\begin{document}
\maketitle

\begin{abstract}
We present TALUS, a threshold signing protocol for ML-DSA (FIPS~204) that supports an
\emph{arbitrary} number of parties, emits signatures accepted by any unmodified ML-DSA verifier, and
reduces, for a bounded number of signatures per key, to ML-DSA's own assumptions (Module-LWE and
SelfTargetMSIS).

Its core is the \emph{Boundary Clearance Condition} (BCC): for a constant fraction of nonces ($31.7\%$ at
ML-DSA-65), the secret vector $\bs_2$ provably cannot cross a rounding boundary, so the $\bs_2$-dependent
rejection check can be enforced offline on preprocessed nonces. This removes the interaction that forces
multi-round signing in prior schemes, yielding a TEE-assisted one-round profile and a fully distributed
honest-majority MPC profile with two online rounds, both built on a Carry Elimination Framework that
computes the shared commitment $\HighBits(\bA\by)$ on secret shares, a step we show is unavoidable.

We prove a lower bound: any FIPS-204-exact threshold scheme revealing a summed (Irwin--Hall) nonce admits
an efficient key-recovery attack after about $2^{30}$ signatures under one key (at ML-DSA-65), so no
$q_s$-independent unforgeability is possible for this class. Complementing it, a $q_s$-bounded analysis
certifies over $105$ bits at the operational signing cap for the external view, with insider ceilings
disclosed per view and a mandatory key rotation capping each key's signing lifetime. A Rust implementation
across all three FIPS~204 levels and a hardware-validated single-party core demonstrate
practicality.\footnote{Single-machine microbenchmarks isolate the arithmetic at low single-digit
milliseconds per signature; deployment wall-clock is dominated by the one or two network round-trips, which
we do not measure. What is proven in this paper versus deferred to the extended version is stated precisely
in \S\ref{sec:intro}.}

\end{abstract}

\section{Introduction}
\label{sec:intro}

Post-quantum migration is turning ML-DSA (FIPS~204~\cite{FIPS204}) into the default signature for
high-assurance deployments: certificate authorities, hardware security modules, and enterprise key
custody. Many of these settings need \emph{threshold} signing (splitting the signing key across $N$
parties so that any $T$ can sign but fewer than $T$ learn nothing) while still emitting signatures
that existing verifiers accept without modification. Threshold ML-DSA is therefore both practically
pressing and technically awkward: unlike Schnorr, ML-DSA's signing loop contains a
\emph{secret-dependent rejection check}, and evaluating it on secret shares is what makes prior
FIPS-compatible threshold schemes multi-round.\footnote{This version supersedes the earlier preprint;
\S\ref{sec:disc} records its relationship to concurrent cryptanalysis~\cite{Niot26}.}

\begin{table}[t]
\centering
\small
\begin{tabular}{lccl}
\toprule
Scheme & Online rounds & Party scale & Assumption \\
\midrule
Quorus~\cite{BdCE25}          & $16$--$29$              & large $N$      & standard, honest maj. \\
Efficient Thr. ML-DSA~\cite{CDENP26} & $3$ / attempt      & $N\le 6$       & standard \\
Trilithium~\cite{Trilithium25}& $14$ / attempt          & $N=2$          & standard \\
\midrule
\textbf{TALUS} (this work)    & \textbf{TEE 1 · MPC 2}  & arbitrary $N$  & standard (MLWE, SelfTargetMSIS) \\
\bottomrule
\end{tabular}
\caption{Online round count and party scale of FIPS-204-compatible threshold ML-DSA schemes. Rounds trade
off against party scale, so the meaningful comparison is the pair. TALUS holds a low online round count at
arbitrary $N$; the security cost of doing so is a $q_s$-bounded guarantee (\S\ref{sec:security}). Beyond the
lattice assumptions shown, TALUS-TEE additionally assumes a trusted enclave and TALUS-MPC an honest majority
($N\ge 2T{-}1$).}
\label{tab:compare}
\end{table}

\paragraph{The observation.}
ML-DSA signs $\bz = \by + c\bs_1$ and accepts only if two norm bounds hold; the one that involves
the second secret vector $\bs_2$ is a check on the low bits $\rzero$ of $\bw=\bA\by$. Prior threshold
schemes evaluate this $\bs_2$-dependent check distributively on every rejection-sampling attempt, and
that distributed evaluation is the source of their round complexity. We observe that the check is
frequently \emph{redundant}: since $\lVert c\bs_2\rVert_\infty \le \beta = \tau\eta$, a coefficient of
$\rzero$ that sits more than $\beta$ from its rounding boundary cannot be moved across the boundary by
\emph{any} admissible $\bs_2$ and challenge $c$. We call the condition that \emph{every} coefficient
clears its boundary by $\beta$ the Boundary Clearance Condition (BCC). The per-coefficient geometry is
exactly the hint-correctness lemma of Dilithium~\cite{DKLLSS18,FIPS204}; our contribution is the two
observations layered on top of it: BCC is \emph{key- and message-independent}, hence checkable in
preprocessing, and enforcing it lets a scheme \emph{delete} the distributed $\bs_2$ \emph{rejection-check}
rather than perform it on shares. The emitted hint still carries a $q_s$-bounded $\bs_2$ residual
(\S\ref{sec:security}), so what is removed is the online $\bs_2$ MPC, not every trace of $\bs_2$. For ML-DSA-65, BCC holds for $\approx 31.7\%$ of nonces, so a usable
nonce is found after $\approx 3.15$ offline trials.

BCC has no benefit for a single signer: there, the $\rzero$ check is free (one resamples $\by$
locally), so pre-selecting nonces only wastes work. The value is intrinsic to the threshold setting,
where that same check otherwise costs a multi-round MPC on $\bs_2$. Offline nonce selection for
single-party Dilithium has been explored~\cite{PreRej24}, but for a different rejection condition and
by filtering the \emph{key} at key generation in a way that changes the key distribution; BCC filters
\emph{nonces}, keeps the exact FIPS-204 nonce range $\gamma_1$ and the key distribution, and eliminates the
online $\bs_2$ computation, though not the $\bs_2$ dependence of the emitted signature, which retains a
$q_s$-bounded residual (\S\ref{sec:security}). The impossibility
of a homomorphic one-round shortcut (that some carry resolution is unavoidable) rests on
$\HighBits$ being non-homomorphic, a fact traceable to~\cite{CozzoSmart19,PMS26} that we formalize for
this setting rather than claim as new.

\paragraph{The cost of exactness.}
Keeping the FIPS-204 signature distribution has a price that we make precise rather than hide. In the
threshold setting the nonce is a \emph{sum} of the parties' shares, so its distribution is
Irwin--Hall, more concentrated than the uniform nonce of single-party ML-DSA. This concentration is
the principal distributional difference from ML-DSA; moving the $\bs_2$-shaping offline adds a second,
$q_s$-bounded departure (the residual $\bs_2$ channel of \S\ref{sec:security}). Neither is free: we show
(\S\ref{sec:security}) that any FIPS-204-exact threshold scheme that \emph{reveals} such a summed nonce leaks
a constant amount of Fisher information about $\bs_1$ per signature, giving an efficient key-recovery attack
after $q_s\approx 4/(I\tau)$ signatures under a single key. No assumption removes this bound; it is intrinsic
to the combination of exact FIPS output and a revealed sum-of-shares nonce. Schemes that avoid it do so
by revealing an internally \emph{uniform} nonce (via more interaction) or by changing the distribution
(non-FIPS); TALUS instead accepts a $q_s$-bounded guarantee and pairs it with a mandatory key \emph{rotation}
(fresh key generation to a new public key) that caps a key's signing lifetime with margin below every wall in
the accounting of \S\ref{sec:security}. We regard this lower bound as a
contribution: it delineates the design space that any FIPS-exact threshold ML-DSA must live in.

\paragraph{Contributions.}
\begin{itemize}[leftmargin=1.4em]
  \item \textbf{The Boundary Clearance Condition (\S\ref{sec:bcc}).} A key- and message-independent
    predicate on the nonce that makes ML-DSA's $\bs_2$-dependent \emph{middle} rejection check redundant,
    moving that $\bs_2$ evaluation (and the online rounds it costs) into preprocessing. We prove the
    redundancy (building on the Dilithium hint lemma), characterize the $\approx 31.7\%$ acceptance rate,
    and disclose the $q_s$-bounded $\bs_2$ residual it leaves in the emitted signature.
  \item \textbf{A lower bound for FIPS-exact threshold ML-DSA (\S\ref{sec:security}).} An efficient,
    assumption-free key-recovery attack showing that any FIPS-204-exact scheme revealing a summed nonce
    is $q_s$-bounded, with the wall at $q_s\approx 4/(I\tau)$. This is the technical boundary separating
    TALUS from schemes (Quorus~\cite{BdCE25}, Trilithium~\cite{Trilithium25}) that reveal an internally uniform nonce.
  \item \textbf{Two practical profiles and a carry framework (\S\ref{sec:tee}--\S\ref{sec:mpc}).} A
    TEE-assisted profile with one online round for any $T$-of-$N$; a fully distributed MPC profile with
    two online rounds under an honest majority (identifiable abort for preprocessing faults, online attribution deferred); and a Carry Elimination Framework
    that computes $\wone=\HighBits(\bA\by)$ on shares without revealing the nonce, together with the
    impossibility result motivating it.
  \item \textbf{Implementation (\S\ref{sec:impl}).} A Rust implementation across all three FIPS~204
    levels and a hardware-validated single-party core, with single-machine compute microbenchmarks. All outputs are
    accepted by an unmodified ML-DSA verifier.
\end{itemize}

\paragraph{Scope of the security claims.}
We are explicit about what is and is not proven. The signatures are byte-identical in format to FIPS-204 outputs
and accepted by any unmodified verifier; the \emph{signing distribution} departs from single-party
ML-DSA in $q_s$-bounded ways: the nonce concentration quantified above, and the residual $\bs_2$ channel
that moving the online $\bs_2$-shaping offline introduces (\S\ref{sec:security}). Security is analyzed directly against ML-DSA's
own assumptions (Module-LWE for key hiding, SelfTargetMSIS for unforgeability), not by treating
single-party ML-DSA as a black box. The summed nonce is not ML-DSA-distributed, so such a reduction
does not apply. The guarantee is $q_s$-bounded as the lower bound forces. EUF-CMA unforgeability is proven
against malicious adversaries with well-formed (extractable) key-generation outputs
(Theorem~\ref{thm:ub}, Remark~\ref{rem:malicious-eufcma}): the game sequence operates
on honest-party randomness, with the insider ceiling $\lambda/m$ for $m$ honest nonce shares.
Simulation-based transcript privacy is proven in the semi-honest model
(Proposition~\ref{prop:semihonest}); its malicious lift (requiring ZK-extractable inputs in the DKG and nonce
generation) and the DKG instantiation are developed in the extended version. We flag each such boundary in
place.

\section{Preliminaries}
\label{sec:prelim}

\subsection{Notation}
Let $q = 8380417$ and $\Rq = \Z_q[X]/(X^{256}+1)$. Bold lower-case letters denote vectors over $\Rq$
and bold upper-case letters matrices; $\infnorm{\cdot}$ is the coefficient-wise $\ell_\infty$ norm.
For a set $S$, $x \getsr S$ is a uniform draw. We write $[k]=\{1,\dots,k\}$. The public matrix is
$\bA \in \Rq^{k\times\ell}$ with $k\ge\ell$; the secret is $(\bs_1,\bs_2)$ with
$\infnorm{\bs_1},\infnorm{\bs_2}\le\eta$; the public key is $\bt=\bA\bs_1+\bs_2$, split as
$\bt=\bt_1\cdot 2^d+\bt_0$. TALUS publishes the full $\bt$ (both $\bt_1$ and $\bt_0$) so that
hint formation is public (\S\ref{sec:impl}); standard ML-DSA keeps $\bt_0$ in the secret key.
We write $n_\ell=256\ell$ and $n_k=256k$ for the number of integer coefficients in
an $\ell$-vector and a $k$-vector over $R_q$.

\subsection{ML-DSA (FIPS 204)}
\label{sec:prelim-mldsa}
We recall the parameters used throughout at security level~3 (ML-DSA-65); levels~2 and~5 are analogous
and stated where they matter.
\begin{itemize}[nosep,leftmargin=1.4em]
  \item $\gamma_1 = 2^{19}$ (nonce range), $\gamma_2 = (q-1)/32 = 261{,}888$ (rounding granularity),
    $\alpha = 2\gamma_2 = (q-1)/16$ (stripe width), $\tau=49$, $\eta=4$, hence $\beta=\tau\eta=196$;
    $(k,\ell)=(6,5)$, $\omega=55$, $d=13$.
\end{itemize}
All three levels are tabulated in Table~\ref{tab:params}; every per-level security number in this paper is
computed from these FIPS-204 parameters.
\begin{table}[h]\centering\small
\caption{FIPS-204 parameters for the three security levels ($q=8{,}380{,}417$, $d=13$ throughout).}
\label{tab:params}
\begin{tabular}{lcccccccccc}
\hline
Level & $(k,\ell)$ & $\gamma_1$ & $\gamma_2$ & $\alpha=2\gamma_2$ & $\tau$ & $\eta$ & $\beta$ & $\omega$ & $n_\ell$ & $n_k$ \\
\hline
ML-DSA-44 & $(4,4)$ & $2^{17}$ & $95{,}232$ & $190{,}464$ & $39$ & $2$ & $78$ & $80$ & $1024$ & $1024$ \\
ML-DSA-65 & $(6,5)$ & $2^{19}$ & $261{,}888$ & $523{,}776$ & $49$ & $4$ & $196$ & $55$ & $1280$ & $1536$ \\
ML-DSA-87 & $(8,7)$ & $2^{19}$ & $261{,}888$ & $523{,}776$ & $60$ & $2$ & $120$ & $75$ & $1792$ & $2048$ \\
\hline
\end{tabular}
\end{table}

\paragraph{Rounding.} $\Decompose(w,\alpha)=(\wone,\rzero)$ writes $w = \wone\cdot\alpha+\rzero$ with
$\rzero\in(-\alpha/2,\alpha/2]$ (with the FIPS-204 boundary convention); $\HighBits(w,\alpha)=\wone$ and
$\LowBits(w,\alpha)=\rzero$. The hint functions $\MakeHint$ and $\UseHint$ let a verifier recover
$\HighBits(w)$ from an approximation, and satisfy the Dilithium hint-correctness property recalled in
\S\ref{sec:bcc}.

\paragraph{Signing (single-party, simplified).}
Sample $\by \getsr [-\gamma_1{+}1,\gamma_1]^{n_\ell}$, set $\bw=\bA\by$ and
$(\wone,\rzero)=\Decompose(\bw,\alpha)$, compute, for the message representative $\mu$, the challenge
$c=\mathsf{SampleInBall}(\Hash(\mu\Vert\wone))$ (a polynomial in $R_q$ with exactly $\tau$ coefficients in
$\{-1,+1\}$ and the rest zero) and the response $\bz=\by+c\bs_1$. \emph{Reject} (resample) unless
\begin{equation}\label{eq:reject}
  \infnorm{\bz} < \gamma_1-\beta,\qquad
  \infnorm{\rzero-c\bs_2} < \gamma_2-\beta,\qquad
  \mathrm{wt}(\bh)\le\omega,
\end{equation}
where the hint is $\bh=\MakeHint(-c\bt_0,\,\bw-c\bs_2+c\bt_0,\,\alpha)$. The signature is
$(c,\bz,\bh)$.

\paragraph{Verification.}
Recover $\wone'=\UseHint(\bh,\,\bA\bz-c\bt_1\cdot 2^d,\,\alpha)$ and accept iff
$\infnorm{\bz}<\gamma_1-\beta$, $\mathrm{wt}(\bh)\le\omega$, and $c=\Hash(\mu\Vert\wone')$.
A signature accepted by this procedure is accepted by any unmodified FIPS-204 verifier.

\paragraph{The secret-dependent check.} Of the three rejection conditions in~\eqref{eq:reject}, the
$\bz$-norm and hint-weight conditions depend on the message (through $c$) and are evaluated online in
every threshold scheme; the middle condition, $\infnorm{\rzero-c\bs_2}<\gamma_2-\beta$, is the one that
touches $\bs_2$, and evaluating it distributively is the source of prior schemes' online round cost.
It is this condition that BCC makes redundant (\S\ref{sec:bcc}).

\subsection{Shamir sharing and the nonce}
\label{sec:prelim-share}
We use standard Shamir~\cite{Shamir79} $(T,N)$ sharing over $\Rq$: $\sthr{x}$ denotes a degree-$(T{-}1)$ sharing, and any
$T$ parties reconstruct $x$ by Lagrange interpolation while $T{-}1$ learn nothing. The signing key
$\bs_1$ is Shamir-shared and never reconstructed during signing.

\paragraph{The summed nonce.} In the distributed profile no party holds the nonce; each contributes a
share, and the effective nonce $\by=\sum_h \by_h$ is summed over the $\nu$ parties that contribute one. Because
$\by$ is a \emph{sum} of $\nu$ independent bounded terms, its per-coefficient law is a \emph{discrete}
Irwin--Hall law: the distribution of a sum of $\nu$ discrete uniforms, not the continuous Irwin--Hall density.
We write $\IH(\gamma_1,\nu)$ for it, keeping the name by analogy. The count $\nu$ is profile-dependent
($\nu\approx T$ signer shares for the TEE profile, $\nu=2T-1$ for the MPC committee; \S\ref{sec:mpc}), and this
law is more concentrated than the uniform law of the single-party nonce. This concentration is the
principal distributional difference between TALUS signatures and single-party ML-DSA signatures; moving the
$\bs_2$-shaping offline adds a second, $q_s$-bounded one (the residual $\bs_2$ channel), and both security
costs are quantified in \S\ref{sec:security}.

\subsection{Security definitions}
\label{sec:prelim-sec}
\begin{definition}[EUF-CMA, threshold]\label{def:eufcma}
A $(T,N)$-threshold signature is EUF-CMA secure against \emph{static} corruption if no PPT adversary that
statically corrupts up to $T{-}1$ parties, and interacts with the honest parties over polynomially many
signing sessions on adaptively chosen messages, produces a valid signature on a fresh message except with
negligible probability. The adversary's view includes the full protocol transcript of every session it
participates in, including any values revealed to corrupt parties before a message is chosen.
\end{definition}
We treat static corruption throughout; adaptive corruption is left open. Definition~\ref{def:eufcma}
explicitly grants the adversary the pre-message transcript, so that a proof under it must handle values
(such as an offline commitment) that are fixed before the message, the setting relevant to
\S\ref{sec:security}.

\paragraph{Assumptions.} Key hiding rests on Module-LWE and unforgeability on the SelfTargetMSIS problem,
the same pair underlying ML-DSA's own security~\cite{DKLLSS18,KLS18,FIPS204}. We state both in
Appendix~\ref{app:assumptions}. We do \emph{not} reduce to single-party ML-DSA as a black box: the summed
nonce is not ML-DSA-distributed, so such a reduction does not apply (\S\ref{sec:security}).

\paragraph{MPC primitives and threat model.} The distributed profile uses standard honest-majority MPC
building blocks (secret-shared multiplication, opening) over $\Rq$; we specify them together with the
carry framework in Appendix~\ref{app:cef}. The MPC profile assumes an honest majority ($N\ge 2T-1$) and provides identifiable
abort on the optimistic path; EUF-CMA unforgeability is proven against malicious adversaries with
well-formed key-generation outputs (\S\ref{sec:security}); simulation-based
transcript privacy is proven in the semi-honest model, with the malicious lift in the extended version
(\S\ref{sec:disc}).

\section{The Boundary: Two Jobs of the Rejection Check}
\label{sec:boundary}

Threshold ML-DSA inherits a single hard constraint from FIPS-204, and every design decision in this paper is a
response to it. We state it first, as one theorem, because it explains what BCC can and cannot buy, why the
security analysis takes the shape it does, and where the honest limit of the construction lies.

\paragraph{The rejection check does two separable jobs.} ML-DSA's middle rejection condition
$\infnorm{\LowBits(\bw-c\bs_2)}<\gamma_2-\beta$ (\S\ref{sec:prelim}) simultaneously performs two logically
distinct tasks:
\begin{enumerate}[nosep,leftmargin=1.6em]
  \item \textbf{correctness}: it certifies that the emitted hint recovers $\wone$, so the signature verifies
    under an unmodified verifier;
  \item \textbf{secrecy}: as a rejection-sampling step it re-centres the accepted transcript so that the
    emitted signature is statistically \emph{independent of $\bs_2$}.
\end{enumerate}
Any threshold scheme that emits FIPS-204 signatures must reproduce job~1. Whether it reproduces job~2 is a
design choice, and it has a price.

\paragraph{A nonce-only predicate can do only correctness.} The Boundary Clearance Condition of
\S\ref{sec:bcc} is a predicate on $(\bA,\by)$ alone; it is blind to $\bs_2$. It can therefore be enforced
offline to certify job~1, but it cannot perform job~2, which is an operation \emph{on $\bs_2$}. Replacing the
online $\bs_2$-rejection by an offline BCC filter thus keeps correctness but drops the $\bs_2$-shaping: each
emitted signature then carries a residual, secret-dependent quantity.

\paragraph{The two channels, both forced by the byte-exact corner.} A byte-exact threshold scheme reveals an
aggregate nonce $\by=\sum_h\by_h$ that is a \emph{sum} of bounded shares, hence Irwin--Hall rather than uniform.
Write $I\approx\tfrac{2}{\gamma_1^2}\ln(\gamma_1/\beta)$ for the $\beta$-truncated Fisher information of that summed
per-coordinate marginal about a location shift; $I>0$ exactly because the sum is non-uniform, and $I=0$ only
when the revealed marginal is uniform (a single term, or an MPC-assembled uniform nonce). Two facts of the
byte-exact threshold design, the revealed non-uniform summed nonce and BCC's dropping of the
$\bs_2$-shaping (job~2), open \emph{two} leakage channels:
\begin{itemize}[nosep,leftmargin=1.6em]
  \item \emph{the $\bs_1$ channel (inherent to any summed nonce).} The response $\bz=\by+c\bs_1$ is a location
    observation of $\bs_1$ with per-coordinate information $I$; a passive, assumption-free estimator recovers
    $\bs_1$ once $q_s\gtrsim 4/(I\tau)$ (\S\ref{sec:lb}, Theorem~\ref{thm:lb}). At ML-DSA-65,
    $I=5.74\times10^{-11}$ and this wall is $\approx 2^{30}$.
  \item \emph{the $\bs_2$ channel (introduced by dropping job~2).} Once the online $\bs_2$-rejection is removed,
    the emitted hint reveals $\mathbf{O}=\rzero+c(\bt_0-\bs_2)$ with $\rzero=\LowBits(\bA\by)$; since
    $\bt_0-\bs_2=\bA\bs_1-\bt_1 2^d$, recovering $\mathbf{O}$ yields $\bA\bs_1$, hence the key. The noise $\rzero$ is
    the BCC-cleared low part: \emph{uniform, with a sharp edge}. A bounded-noise estimator exploits that edge at
    a rate linear (not quadratic) in the noise range, giving a per-key wall $\approx 2\gamma_2$:
    $2^{17.5}/2^{19}/2^{19}$ at ML-DSA-44/65/87. Since $\bt_0$ is published, $\mathbf{O}$ is directly
    computable from public data ($\mathbf{O}=\bA\bz-c\bt_1 2^d-\alpha\wone$), so the sharp-edge observable is
    the actual channel, not an idealization.
\end{itemize}
The two channels are governed by \emph{different} constants: the $\bs_1$ wall by the nonce Fisher
information $I$, the $\bs_2$ wall by the range $2\gamma_2$ of the BCC-cleared low part (FIPS-fixed). But they
are governed by a \emph{single cause}: keeping the output byte-exact forces both a revealed summed nonce ($I>0$) and a BCC
that cannot reproduce the $\bs_2$-shaping. Reshaping the nonce (wider or Gaussian) would relax the $\bs_1$
side by lowering $I$, but forfeits byte-exactness and leaves the FIPS-fixed $\gamma_2$ untouched. \emph{Two
prices, one corner.}

\begin{theorem}[The FIPS-exact threshold boundary; informal]\label{thm:boundary}
Let a threshold scheme emit FIPS-204-exact ML-DSA signatures whose revealed nonce is a sum of $\ge 2$ bounded
shares (Fisher information $I>0$). Then its per-key security is capped: the $\bs_1$ channel alone forces
$q_s\lesssim 4/(I\tau)$ ($\approx 2^{30}$ at ML-DSA-65), and if the scheme additionally omits the
$\bs_2$-output-shaping the $\bs_2$ channel adds a per-key wall of order $\gamma_2$ ($\approx 2\gamma_2$ for
the sharp-edge observable, directly computable from public data when $\bt_0$ is published). No such scheme
has $q_s$-independent unforgeability under any assumption.
\end{theorem}

\paragraph{A trilemma, not a defect.} Theorem~\ref{thm:boundary} says a byte-exact threshold ML-DSA cannot
have all three of \{FIPS-204-exact output, an efficient online phase, an uncapped/tight reduction\}: it may
have any two. Schemes that reveal an internally uniform nonce (Quorus~\cite{BdCE25}, Trilithium~\cite{Trilithium25})
buy $I=0$ with heavier interaction; flooding schemes (Raccoon lineage~\cite{Raccoon2024}) buy it by reshaping
the nonce, and are not FIPS-exact. TALUS keeps byte-exactness \emph{and} an efficient online phase, and pays
with a $q_s$-capped reduction. We regard characterising that price precisely as a contribution
in its own right: an explicit attack (\S\ref{sec:lb}), an optimized upper bound whose remaining distance to
the attack is an exact per-level constant we argue no black-box reduction closes (\S\ref{sec:ub}), and an
honestly placed operational cap (\S\ref{sec:refresh-cap}). The
rest of the paper does exactly that.

\section{The Boundary Clearance Condition}
\label{sec:bcc}

This section defines BCC, proves that it makes the secret-dependent rejection check redundant, and
records its acceptance rate. The per-coefficient geometry is the Dilithium hint-correctness lemma; the
contribution is the observation that the resulting condition is key- and message-independent, hence
enforceable offline, and that enforcing it lets a threshold scheme delete the distributed $\bs_2$
computation rather than perform it.

\subsection{Definition}
\begin{definition}[Boundary Clearance Condition]\label{def:bcc}
Fix public $\bA$ and a nonce $\by$, and let $\rzero=\LowBits(\bA\by,\alpha)$. The nonce satisfies $\BCC$ if
\[
  \infnorm{(\rzero)_j} < \gamma_2-\beta \quad\text{for every coefficient } j,
  \qquad \beta=\tau\eta .
\]
\end{definition}
$\BCC$ is a predicate on $(\bA,\by)$ alone: it does not depend on the secret key $(\bs_1,\bs_2,\bt_0)$ or
on the message $\mu$. This is the property that makes it useful, and we return to it in
\S\ref{sec:bcc-independence}.

\subsection{Redundancy of the secret-dependent check}
The middle rejection condition of \eqref{eq:reject} is $\infnorm{\rzero-c\bs_2}<\gamma_2-\beta$, equivalently
$\HighBits(\bw-c\bs_2)=\HighBits(\bw)=\wone$. The following is the standard Dilithium hint-correctness
lemma~\cite{DKLLSS18,FIPS204}, restated for our use.

\begin{lemma}[Clearance implies invariance]\label{lem:clearance}
Let $\bw=\bA\by$ and $\rzero=\LowBits(\bw,\alpha)$. If $\infnorm{(\rzero)_j}<\gamma_2-\beta$ for all $j$, then
for every admissible challenge $c$ (with $\tau$ nonzero $\pm1$ coefficients) and every secret with
$\infnorm{\bs_2}\le\eta$,
\[
  \HighBits(\bw-c\bs_2,\alpha)=\HighBits(\bw,\alpha).
\]
\end{lemma}
\begin{proof}
Coefficient-wise, $\infnorm{c\bs_2}\le\tau\eta=\beta$, so each coefficient of $c\bs_2$ moves the
corresponding coefficient of $\bw$ by at most $\beta$. Under $\BCC$ the low part $(\rzero)_j$ is more than
$\beta$ from either stripe boundary $\pm\gamma_2$, so adding a shift of magnitude $\le\beta$ cannot cross a
boundary; the high part is unchanged. This is exactly the argument by which the Dilithium hint is correct;
we use it in the opposite direction, to certify when the low-bits check is automatically satisfied.
\end{proof}

\begin{corollary}[The $\bs_2$-check is redundant under BCC]\label{cor:redundant}
For any BCC-passing nonce, $\HighBits(\bw-c\bs_2)=\wone$ for every admissible $c,\bs_2$
(Lemma~\ref{lem:clearance}), equivalently $\infnorm{\rzero-c\bs_2}<\gamma_2$. A scheme that signs only with
BCC-passing nonces therefore substitutes BCC for this middle rejection condition and never evaluates it, and
$\bs_2$ affects neither the high bits $\wone$, the challenge $c=\Hash(\mu\Vert\wone)$, nor the response
$\bz=\by+c\bs_1$. This is \emph{not} the ML-DSA norm condition $\infnorm{\rzero-c\bs_2}<\gamma_2-\beta$:
TALUS's middle accept-region is a $c\bs_2$-shifted variant of ML-DSA's, accepting a thin $\beta$-shell that
ML-DSA rejects and rejecting one it accepts; this shell is the source of the residual $\bs_2$-channel of
\S\ref{sec:security}. BCC removes only this middle check; $\bs_2$ still enters the online hint-weight
rejection.
\end{corollary}
Corollary~\ref{cor:redundant} is what removes $\bs_2$ from the online protocol: the one rejection condition
that a threshold scheme would otherwise evaluate on secret shares is, on the BCC-passing subset, satisfied
unconditionally. This redundancy concerns the \emph{online computation of $\bs_2$}, which BCC removes; it is
not a claim of $\bs_2$-secrecy for the emitted signature. The hint still carries a residual
$\bs_2$-dependence, a $q_s$-bounded leakage channel accounted for in \S\ref{sec:security}. Thus BCC provides
hint-correctness, not $\bs_2$-independence of the output.

\subsection{Key- and message-independence, and where the value lies}
\label{sec:bcc-independence}
$\BCC$ depends only on $\bA\by$. Both inputs are independent of the long-term secret: $\bA$ is public, and
$\by$ is a fresh ephemeral nonce, uncorrelated with $(\bs_1,\bs_2)$. Hence the BCC bit reveals nothing about
the key, and, being independent of $\mu$, can be evaluated before any message is known. This is what lets
BCC filtering run in preprocessing. We emphasize that $\by$ itself is secret; it is BCC's independence from
the key and message, not any publicity of $\by$, that makes it a sound offline filter.

\paragraph{BCC is a predicate on the aggregate, evaluated privately.} A point that matters in the
distributed profile: $\BCC$ is a predicate on the \emph{aggregate} nonce $\by=\sum_h\by_h$, not on any
party's individual share. It is not, and cannot be, checked per-party: $\rzero=\LowBits(\bA\by)$ is a
non-linear function of $\by$, so $\BCC(\by_h)$ says nothing about $\BCC(\sum_h\by_h)$. Since no party holds
$\by$, and $\by$ is not known ahead of time, offline BCC filtering must be a \emph{privacy-preserving} secure
computation of the aggregate BCC bit on shares, revealing only the single (key-independent) pass/fail bit and
never $\by$, $\bw$, or $\rzero$. This is precisely what the Carry Elimination Framework provides
(\S\ref{sec:carry}); the enclave profile instead evaluates the same predicate on the aggregate held
in the clear inside the enclave (\S\ref{sec:tee}).

\paragraph{No single-party incentive.} $\BCC$ offers no cost benefit to a single signer: there the
$\rzero$-check is a local, essentially free comparison, so pre-selecting nonces to avoid it only wastes the
$\approx\!3\times$ oversampling below. The value is intrinsic to the threshold setting, where that same
check costs a distributed, multi-round MPC on $\bs_2$. Offline nonce selection for single-party Dilithium
has been studied~\cite{PreRej24}, but it filters the \emph{key} at key generation for a different rejection
condition and changes the key distribution; BCC filters \emph{nonces}, keeps the exact FIPS-204 nonce range
$\gamma_1$ and the key distribution, and eliminates the $\bs_2$ computation. This asymmetry (value only under threshold) is the
reason the observation has been available yet unused.

\subsection{Acceptance rate}
\begin{proposition}[BCC rate]\label{prop:bccrate}
For a uniform single-coefficient low part, $\Pr[\,|(\rzero)_j|<\gamma_2-\beta\,]=1-\beta/\gamma_2$, and over
$n_k$ independent coefficients the BCC acceptance probability is
$p_{\BCC}=(1-\beta/\gamma_2)^{n_k}$.
\end{proposition}
Numerically $p_{\BCC}=31.7\%$ at ML-DSA-65, and $43.2\%$ / $39.1\%$ at levels~2 / 5, matching a Monte-Carlo
estimate to within $0.3\%$. A usable nonce is therefore found after $1/p_{\BCC}\approx 3.15$ offline trials
at level~3. \emph{These are offline nonce trials, done before any message is known; they are not online
interaction rounds}, a distinction we keep throughout (the online round count is stated separately in
\S\ref{sec:tee}--\S\ref{sec:mpc}).

\section{Carry Elimination}
\label{sec:carry}

BCC removes the $\bs_2$ computation, but the parties still need the commitment $\wone=\HighBits(\bA\by)$ to
form the challenge $c=\Hash(\mu\Vert\wone)$. In the distributed profile no party holds $\by$, so $\wone$
must be computed on secret shares. This section gives that computation, the Carry Elimination Framework
(CEF), and shows that some form of carry resolution is unavoidable. We begin with a protocol overview that
ties together the full signing lifecycle (detailed in the subsections that follow and in the algorithms of
Appendix~\ref{app:cef}).

\begin{algorithm}[t]
\caption{Protocol overview: TALUS signing lifecycle (MPC profile; TEE differences noted).}\label{alg:overview}
\begin{algorithmic}[1]
\Statex \textbf{Phase 0 --- Key generation} (once per key; DKG details in the extended version)
\State dealer-free DKG $\to$ Shamir shares $\sthr{\bs_1},\sthr{\bs_2}$; publish $\pk=(\rho,\bt_1,\bt_0)$
  with $\bt=\bA\bs_1+\bs_2$
\Statex
\Statex \textbf{Phase 1 --- Preprocess} (offline, repeated; Alg.~\ref{alg:preprocess}--\ref{alg:cef})
\State nonce DKG: each committee member $h$ samples $\hat\by_h$; Shamir shares sent to all $N$ parties give $\sthr{\by}$
\State each $h$: $\hat\bw_h\gets\bA\hat\by_h$;\
  $(w_{1,h},r_{0,h})\gets\Decompose(\hat\bw_h,\alpha)$
\State \textbf{CEF} (Alg.~\ref{alg:cef}): full-stripe masked broadcast of low parts $\to$ carry
  $\kappa$; combine with shared high parts $\to$ $\wone=\HighBits(\bA\by)$, exact
  \Comment{$\by,\rzero$ never opened}
\State \textbf{BCC filter}: on shares, check $\infnorm{(\rzero)_j}<\gamma_2-\beta$; open one bit;
  discard if fail ($\approx 69\%$ at level~3; \S\ref{sec:bcc})
\State \textbf{[TEE]} enclave holds $\by$ in the clear; computes $\wone$ and BCC directly
\State commit: $\mathsf{com}\gets\Hash(\wone\Vert r)$; pool $(\wone,\mathsf{com},r,\{\hat\by_i\}_{i\in[N]})$
\Statex
\Statex \textbf{Phase 2 --- Sign}$(\mu)$ (online; Alg.~\ref{alg:sign})
\State \textbf{[MPC]} Round 1: open $\mathsf{com}$; all learn $\wone$
\State all signers: $c\gets\mathsf{SampleInBall}(\Hash(\mu\Vert\wone))$
\State Round 2: each signer $h$ sends $\bz_h\gets\hat\by_h+c\cdot\bs_{1,h}$
\State coordinator: $\bz\gets\sum_h\lambda_h\bz_h$; $\bh\gets\MakeHint(\ldots)$;
  verify; emit $\sigma=(c,\bz,\bh)$
\State \textbf{[TEE]} Rounds 1--2 collapse: enclave releases $\wone$ and aggregates $\bz$ internally;
  only accepted $\sigma$ leaves
\end{algorithmic}
\end{algorithm}

\subsection{What is easy and what is hard}
The public matrix $\bA=\mathsf{ExpandA}(\rho)$ is derived from the public $\rho$, so it is known to all
parties (and is injective with overwhelming probability over $\rho$; see Appendix~\ref{app:trilemma}).
Consequently $\bw=\bA\by=\bA\sum_h\by_h=\sum_h\bA\by_h$ is \emph{linear} in the shares: each party
locally forms $\bA\by_h$, and the additive combination gives a sharing $\sthr{\bw}$ of the aggregate with no
interaction. What is \emph{not} linear is $\Decompose(\bw)$: $\HighBits$ does not distribute over the sum,
\[
  \textstyle\sum_h \HighBits(\bA\by_h)\ \ne\ \HighBits(\sum_h \bA\by_h)
\]
in general, because the low parts $\LowBits(\bA\by_h)$ can sum past a stripe boundary and carry into the high
part. Resolving that carry on shares, without revealing the low sum, is the technical core.

\begin{proposition}[Carry--privacy]\label{prop:carrypriv}
Revealing the aggregate low part $\sum_h\LowBits(\bA\by_h) \bmod\alpha$ (equivalently the exact $\rzero$) is
unsafe: from a single signature it yields a noise-free lattice sample, from which $\bs_1$ is recovered by
linear algebra. Hence the low sum must never be opened \emph{in the clear}. CEF instead opens only a
\emph{full-stripe--masked} low sum $B=\sum_h(r_{0,h}+\rho_h)$, each $\rho_h$ drawn uniformly on the whole
stripe $[0,\alpha)$, so that $B\bmod\alpha$ is a \emph{one-time pad} and the opening leaks nothing about
$\rzero$ beyond the unmasked coarse quotient $\lfloor B/\alpha\rfloor$; only the key-independent stripe-carry,
modulus-wrap, and FIPS boundary bits are then extracted (on shares, never opened). That residual quotient has committee-independent per-key wall
$\ge 2^{28.2}$ across levels ($2^{28.2}/2^{30.8}/2^{30.5}$), above the $\bs_1$ wall, hence non-binding
(\S\ref{sec:security}, Appendix~\ref{app:cef}); we disclose it rather than treat it as zero.
\end{proposition}

\subsection{The framework}
A modulus identity of ML-DSA makes the carry cheap. Because $q-1$ is an exact multiple of the stripe width
($q=m\alpha+1$), the high part takes one of
$m=(q-1)/\alpha$ values ($m=16$ at levels~3 and~5; $m=44$ at level~2, handled by parameterizing the
reduction). With the full-stripe masks of Proposition~\ref{prop:carrypriv} the masked low sum can wrap
several stripes, so the secure carry is $\lceil\log_2|C|\rceil+1$ bits per coefficient (how many stripe
boundaries the masked sum crossed) rather than the single bit a magnitude-bounded mask would give; this
widening is exactly what makes the mask a true one-time pad while keeping $\wone$ byte-exact. CEF computes,
on shares:
\begin{enumerate}[nosep,leftmargin=1.6em]
  \item the shared low-sum in carry-save form, without opening it;
  \item the aggregate stripe carry and modulus-wrap (the multi-bit $\kappa^\rho$, the sub-stripe borrow $b$,
    and the mod-$q$ wrap count $k$; Appendix~\ref{app:cef}) by secure comparisons on the shares;
  \item the FIPS round-to-nearest centering bit $\delta$ for the centered representative,
\end{enumerate}
and combines them to obtain $\sthr{\wone}=\HighBits(\bA\by)$ on shares. The round and communication cost is
$O(\log N)$ per signature and is confined to preprocessing; the full protocol, including the $T=2$
specialization and the $m=44$ case, is in Appendix~\ref{app:cef}.

\paragraph{BCC is a by-product of the same computation.} Once CEF has produced $\sthr{\wone}$, the low part
follows linearly: $\sthr{\rzero}=\sthr{\bw}-\alpha\cdot\sthr{\wone}$. The aggregate BCC bit is then the
conjunction $\bigwedge_j\big[\,|(\rzero)_j|<\gamma_2-\beta\,\big]$, evaluated by the same secret-shared
comparisons; only that single (key-independent) bit is opened, and a failing candidate is discarded.
Thus one carry resolution yields both the commitment $\wone$ (used online) and the offline BCC filter; the
exact low sum and $\by$ are never opened in the clear. Only the masked (fuzzy) $\rzero$ is, per
Proposition~\ref{prop:carrypriv}, and the full product $\bw$ is never reconstructed.
Because the opened bit is key-independent it leaks nothing about the key per candidate; the residual
information from \emph{which} candidates fail, accumulated over many signatures, is exactly the
$q_s$-bounded regime of \S\ref{sec:security} and is capped by the mandatory key rotation.

\subsection{Unavoidability}
CEF adds an offline carry resolution; the following says no FIPS-compatible scheme can avoid some such step
by a one-round homomorphic shortcut. We state it here and prove it in Appendix~\ref{app:trilemma}.

\begin{theorem}[Carry resolution is necessary; informal]\label{thm:carry-necessary}
There is no one-round protocol in the plaintext linear-aggregation model (parties broadcast a linear
image $\varphi(\by_h)$ and everyone computes $\wone$ from $\sum_h\varphi(\by_h)$ by a public map, with no
trusted reconstruction, no extra rounds, and no rejection) that computes $\wone=\HighBits(\bA\by)$ while
hiding $\by$. Any FIPS-compatible lattice threshold scheme must therefore use one of: trusted/enclave or
MPC reconstruction of the high bits, extra interaction, or a rejection step.
\end{theorem}
The mechanism is that $\HighBits$ is not a homomorphism (a fact we attribute to
\cite{CozzoSmart19,PMS26}); the impossibility is a formalization of this folklore at the scope the result
needs, not a claim that the non-homomorphism is new. It is unconditional for a wide nonce, and holds with
overwhelming probability for ML-DSA's boxed nonce under a mild expansion hypothesis (Appendix~\ref{app:trilemma}).

\section{TALUS-TEE}
\label{sec:tee}

The TEE profile establishes the one-round online phase using a trusted execution environment as the trust
anchor, for any $T$-of-$N$ and with no honest-majority requirement.

\paragraph{Protocol.} In preprocessing, an enclave coordinator receives the parties' nonce shares, forms the
aggregate nonce $\by$ internally, computes $\wone$ and checks BCC directly (with $\by$ in cleartext inside
the enclave), and discards nonces that fail BCC. It thereby holds a pool of BCC-passing nonces with their
$\wone$. To sign $\mu$: the enclave releases $\wone$, the parties compute $c=\Hash(\mu\Vert\wone)$ and
broadcast $\bz_h=\by_h+c\bs_{1,h}$, and the enclave aggregates $\bz$, forms the hint, verifies, and emits
$\sigma$ only if verification passes. This is a single online broadcast round.

\paragraph{Cost and guarantees.} Online latency is one round; measured single-box compute $\approx 1.5$~ms at
ML-DSA-65, $T=3$ (\S\ref{sec:impl}). Because BCC is enforced offline, the $\rzero$ check never fires online; a first
signing attempt succeeds with probability $\approx 99.6\%$ at ML-DSA-65 (measured over $12{,}645$ attempts,
Appendix~\ref{app:empirical}; $\approx 98.0\%$ / $99.2\%$ at levels~2 / 5). The residual comes \emph{entirely} from
ML-DSA's message-dependent hint-weight check: across our $\approx 33{,}000$ measured attempts the $\bz$-norm
check never fired, the summed nonce being concentrated well inside the norm bound (the same Irwin--Hall
concentration quantified in \S\ref{sec:security}). Neither residual can be pre-filtered offline, so the
online phase is one round per attempt with expected $\approx 1.004$ attempts at level~3, not a guaranteed
single pass.

\paragraph{Trust and security.} The enclave is the coordinator; the $\bz_h$, the aggregate, and the norm
check are enclave-internal, and only an accepted $\sigma$ leaves. An adversary corrupting up to $T{-}1$
signers (not the enclave) never observes an aggregate or a rejected $\bz$; its view is $\wone$ and
$\sigma$, simulatable from $(\pk,\sigma)$ up to the $q_s$-bounded channels of \S\ref{sec:security}; for a
coalition of signers, subtracting its own contributions from $\sigma$'s response reduces the summed nonce to
the conditional residual of the insider regime in Appendix~\ref{app:ub} (and the
emitted $\sigma$ carries the summed nonce, so Theorem~\ref{thm:lb} and the rotation cap apply to the TEE
profile as well). The enclave must be trusted for both confidentiality and integrity: it holds the aggregate
nonce $\by$ \emph{and} aggregates $\bz=\by+c\bs_1$, so a confidentiality breach alone recovers
$\bs_1=c^{-1}(\bz-\by)$ by a single linear solve (a full key compromise, not a reduction to
Module-LWE), and an integrity breach can violate the $\wone$-after-$\mu$ ordering or the verify-before-emit
guard. The TEE profile thus rests on the enclave assumption plus the standard assumptions of
\S\ref{sec:security}; under those, unforgeability follows from the $q_s$-bounded reduction of Theorem~\ref{thm:ub}.

\section{TALUS-MPC}
\label{sec:mpc}

The distributed profile removes the enclave. No party holds the nonce, so $\wone$ is computed on shares by
CEF and BCC is filtered on shares; the online phase is two rounds. This section gives the structure, the
principle that determines the round count, and the closure conditions for malicious security.

\subsection{Compute offline, reveal online}
The challenge $c=\mathsf{SampleInBall}(\Hash(\mu\Vert\wone))$ needs both $\mu$ (online) and $\wone$.
Crucially, \emph{computing} $\wone=\HighBits(\bA\by)$ is message-independent and is done offline (by CEF, on
shares); only \emph{revealing} $\wone$ is timing-sensitive. This split is the whole design.

\paragraph{What must stay hidden.} Since $\bA$ is injective, $\bw=\bA\by$ determines $\by$; revealing $\bw$
would reveal $\by$, and with the released $\bz=\by+c\bs_1$ that gives $\bs_1$. So the full product $\bw$ and
its low part $\rzero$ are \emph{never} opened; only the massively lossy high part $\wone=\HighBits(\bw)$ is
released (and only online). CEF's job is exactly to extract the shared $\sthr{\wone}$ from the shared
$\sthr{\bw}$ without ever reconstructing $\bw$ or $\rzero$ (\S\ref{sec:carry}).

\paragraph{Offline.} Parties run a nonce DKG to obtain $\sthr{\by}$ (each party $h$'s free term is
$\hat\by_h$; we write $\by_h$ when the hat is clear from context) and locally form $\sthr{\bw}=\sum_h
\bA\hat\by_h$; CEF resolves the carry on shares by masked broadcast: each party masks its low part with a mask
$\rho_h$ drawn uniformly on the \emph{full stripe} $[0,\alpha)$, and the masked aggregate low part
$B=\sum_h(r_{0,h}+\rho_h)$ is opened. Because each mask spans the whole stripe, $B\bmod\alpha$ is a
\emph{one-time pad} on $\rzero$ (it reveals nothing about the true low sum), and only the key-independent
stripe-carry and FIPS boundary bits, computed securely from the shared masks, are extracted
(Proposition~\ref{prop:carrypriv}; the carry is correspondingly widened to $\lceil\log_2|C|\rceil+1$ bits and
$\wone$ stays byte-exact). This full-stripe masking makes $B\bmod\alpha$ a one-time pad, so the residual is
confined to the coarse quotient $\lfloor B/\alpha\rfloor$, with committee-independent per-key wall
$\ge 2^{28.2}$ across levels ($2^{28.2}/2^{30.8}/2^{30.5}$; \S\ref{sec:security}), well above the cap. To keep the two-round guarantee, $\wone$ is not opened offline:
it is placed under an offline hiding-and-binding commitment (\S\ref{sec:mpc-comm}), so the pool holds
BCC-passing nonces each with a \emph{committed} $\wone$, opened only online. (Producing a genuinely
secret-shared $\sthr{\wone}$, rather than a committed reconstruction opened in Round~1, is a construction
refinement over the reference implementation, which reconstructs $\wone$; both keep $\wone$ hidden until
after $\mu$, which is what the security argument needs.) The on-shares BCC filter opens a single
key-independent pass/fail bit and discards failing nonces; about $3.15$ candidates are sifted per pooled
nonce, in batch ahead of any message.

\paragraph{Online (two rounds).}
\begin{center}\itshape
Round 1: open the committed $\sthr{\wone}$ to reveal $\wone$; each signer locally computes
$c=\mathsf{SampleInBall}(\Hash(\mu\Vert\wone))$.\quad
Round 2: reveal $\bz_h=\by_h+c\bs_{1,h}$; aggregate $\bz=\sum_h\bz_h$, form the hint (a local public step,
since $\bt_0$ is published; \S\ref{sec:impl}), verify, emit $\sigma$.
\end{center}
CEF is not run online; the two online rounds are the two reveals, with the $c$ computation a local step
between them. The final $\sigma$ is self-verifying, so a wrong aggregate is caught before emission.
Algorithm~\ref{alg:sign} gives the pseudocode; in the TEE profile Round~1 is absent (the enclave holds
$\wone$ in cleartext) and all aggregation is enclave-internal.

\begin{algorithm}[t]
\caption{$\mathsf{TALUS.Sign}$: online signing (both profiles).}\label{alg:sign}
\begin{algorithmic}[1]
\Statex \textbf{Input:} message $\mu$; from the pool of Alg.~\ref{alg:preprocess}: $(\wone,\mathsf{com},r)$
  and the quorum's shares $\{\hat\by_h\}_{h\in S}$; signing quorum~$S$, $|S|=T$.
\Statex \textbf{Output:} FIPS-204 signature $\sigma=(c,\bz,\bh)$ or $\bot$.
\State \textbf{[MPC only]} signers open $\mathsf{com}$; all learn $\wone$ \Comment{Round~1}
\State each signer $h\!\in\!S$: $c\gets\mathsf{SampleInBall}\bigl(\Hash(\mu\Vert\wone)\bigr)$ \Comment{local}
\State each signer $h$ sends $\bz_h\gets\hat\by_h+c\cdot\bs_{1,h}$ to coordinator \Comment{Round~2 (Round~1 for TEE)}
\State $\bz\gets\sum_{h\in S}\lambda_h\,\bz_h$ \Comment{Lagrange at $0$: $\lambda_h=\prod_{j\in S\setminus\{h\}}j/(j{-}h)$}
\State \textbf{if} $\infnorm{\bz}\ge\gamma_1-\beta$ \textbf{then return} $\bot$
  \Comment{$\bz$-norm check}
\State $\bh\gets\MakeHint(-c\bt_0,\;\bA\bz-c\bt_1\cdot 2^d,\;\alpha)$
  \Comment{$\bt_0$ is public (\S\ref{sec:impl})}
\State \textbf{if} $\operatorname{wt}(\bh)>\omega$ \textbf{then return} $\bot$
  \Comment{hint-weight check}
\State $\wone'\gets\UseHint(\bh,\;\bA\bz-c\bt_1\cdot 2^d,\;\alpha)$
\State \textbf{verify} $c\stackrel{?}{=}\Hash(\mu\Vert\wone')$; \textbf{if not, return} $\bot$
  \Comment{self-verification}
\State \Return $\sigma=(c,\bz,\bh)$ \Comment{accepted by any unmodified FIPS-204 verifier}
\end{algorithmic}
\end{algorithm}

\subsection{Why two rounds, and why standard assumptions}
\label{sec:mpc-why}
The timing of the $\wone$ reveal is exactly what separates a secure standard-assumption protocol from an
insecure one.
\begin{itemize}[nosep,leftmargin=1.4em]
  \item If $\wone$ is revealed \emph{offline}, before $\mu$ (a one-round online phase), then corrupt parties
    hold the offline commitments before choosing messages, the commit-before-message setting underlying
    ROS/Wagner attacks on rejection-based threshold signatures~\cite{Wagner02,Drijvers19,BLLOR21}. A standard
    reduction cannot then program the challenge, and security would need a one-more assumption.
  \item If $\wone$ is kept \emph{committed/shared until online} and \emph{opened} only after $\mu$ (the
    two-round phase above), there is no pre-message \emph{reveal}, and security rests on standard Module-LWE
    and SelfTargetMSIS (\S\ref{sec:security}).
\end{itemize}
This is the same axis that makes the TEE profile one round: the enclave holds the aggregate in the clear, so
it can apply the BCC filter and be ready to release $\wone$ immediately once $\mu$ arrives, collapsing the
reveal into the single online round, whereas the MPC profile has no such trusted party and must run the
commit-then-open across the two rounds. In both profiles $\wone$ is revealed only after $\mu$; what the
enclave saves is the extra online round the distributed open would otherwise cost. TALUS-MPC adopts the two-round ordering: keeping the
headline standard-assumption is worth one online round. The one-round MPC variant (secure only under a
one-more lattice assumption, in the manner of FROST~\cite{KG20} under OMDL~\cite{BCKMTZ22}) is not pursued here, as its required
structure conflicts with FIPS-exactness and arbitrary~$N$.

\subsection{Honest majority is for the carry, not the round count}
CEF is a non-linear computation on shares and requires an honest majority; the linear online reconstruction
of $\bz$ tolerates $T{-}1$ corruptions. These are independent axes: the two-round ordering fixes the
\emph{assumption} (removes the one-more requirement), while honest majority is what makes the on-shares
$\wone$ computation possible at all. To place the honest-majority guarantee correctly, CEF is run over a
committee of $\ge 2T-1$ parties, so a corrupt minority is a minority \emph{within the committee}, decoupled
from the size-$T$ online signing quorum. The nonce is carried in two coexisting representations from the
same committee polynomials: an additive form over the committee (for the offline $\wone$) and a
degree-$(T{-}1)$ Shamir form (for the online $\bz$), which are numerically equal by linearity of Shamir
interpolation ($\sum_h g_h(0)$ is reconstructed from the shares $\{\sum_h g_h(i)\}$). Because the committee's
nonce polynomials are evaluated at all $N$ points (Alg.~\ref{alg:preprocess}), any $T$ of the $N$ parties hold
both a nonce share and a key share at their evaluation point, so the $T$-of-$N$ signing threshold is unchanged;
the requirement added is offline liveness of the committee.

\subsection{Communication model and identifiable abort}
\label{sec:mpc-comm}
The round counts above are stated in the standard threshold-signature model of one broadcast per round, as
in FROST~\cite{KG20}. We do \emph{not} require a heavyweight Byzantine broadcast channel: the parties
communicate through an \emph{untrusted} aggregator (coordinator), exactly the FROST topology. An untrusted
aggregator could try to equivocate (send different $\wone$ to different signers, so they compute different
$c$), but the offline \emph{hiding-and-binding commitment} to $\wone$ binds it to a single value, so a
divergent opening is detected (a lightweight echo of the opened value across signers suffices for delivery
agreement; a full Byzantine broadcast is not needed). The same commitment removes the rushing/adaptive-$c$
concern (\S\ref{sec:mpc-malicious}). A wrong aggregate $\bz$ fails the pre-emission verification, so the
coordinator cannot substitute a bad signature; note, however, that self-verification blocks a \emph{forged}
signature but not the coordinator-triggered Blame attack, which is the dominant malicious break and is closed
separately in \S\ref{sec:mpc-malicious}.

\paragraph{Identifiable abort.} On the optimistic (fault-free) path, identifiable abort costs \emph{no extra
online round}: attribution rides on the authenticated/committed messages already sent in the two rounds. Only
when a fault is detected does a \emph{Blame subprotocol} run, adding a round on that path alone. Crucially,
Blame attributes faults from the \emph{authenticated offline transcript}, before the online round, and must
never open a post-$\bz$ nonce; otherwise it becomes the key-recovery of \S\ref{sec:mpc-malicious}. Offline
faults (bad preprocessing contributions) are attributable this way. A purely online fault, a signer sending
a malformed $\bz_h$, is the residual case: since a public per-party check would require publishing
$\bA\by_h$ (which must stay hidden), attributing it needs a per-party verifiable-$\bz_h$ proof; we defer that
component, and the complete identifiable-abort treatment, to the extended version. A broadcast-style
agreement may be needed for dispute resolution inside Blame, but only on that failure path, never in the
optimistic two rounds.

\subsection{Malicious model: the closure conditions}
\label{sec:mpc-malicious}
EUF-CMA unforgeability is proven against malicious adversaries with well-formed key-generation outputs
(\S\ref{sec:security},
Remark~\ref{rem:malicious-eufcma}); simulation-based transcript privacy is proven in the semi-honest model
(Proposition~\ref{prop:semihonest}). Identifiable abort on the optimistic path is provided as above; the
simulation-based malicious lift, including online malformed-$\bz_h$ attribution, is deferred to the extended
version. We state here the protocol mechanisms it requires, because they
determine the construction. Every catastrophic malicious break we are aware of has the same shape: two
individually key-independent revealed values whose \emph{join} determines the key. Each is closed by a
standard identifiable-abort or threshold technique that severs one of the two reveals:
\begin{itemize}[nosep,leftmargin=1.4em]
  \item \textbf{Post-$\bz$ nonce reveal.} A coordinator holding $\bz=\by+c\bs_1$ that could also trigger a
    Blame opening $\by$ would recover $\bs_1=c^{-1}(\bz-\by)$. Closure: attribute preprocessing faults
    \emph{before} the online round, via authenticated offline transcripts, so Blame never opens a post-$\bz$
    nonce.
  \item \textbf{Round-1 reveal binding.} A bare open of $\wone$ is rushing-steerable, which would restore
    adversary-chosen $c$ and re-introduce the one-more setting. Closure: an offline hiding-and-binding
    commitment to the $\wone$-shares, opened online; this costs no extra online round, since $\wone$ is
    message-independent.
  \item \textbf{CEF integrity and key generation.} CEF uses authenticated (rather than plain
    preprocessing-based) multiplication; key generation is a dealer-free DKG whose only public lattice output
    is $\bt=\bA\bs_1+\bs_2$ (an MLWE sample, hence hiding) and \emph{never} a per-party $\bA\bs_{1,h}$: that
    is a bare, error-free lattice image, and since $\bA$ has full column rank at ML-DSA's dimensions its
    left-inversion would expose the share $\bs_{1,h}$. Short-share well-formedness is therefore proved in zero
    knowledge rather than by a Feldman $\bA\bs_{1,h}$ commitment. Such a DKG is realizable in the manner of
    Quorus~\cite{BdCE25}; not publishing the non-hiding image closes the passive-share-recovery route by
    construction, while the rogue-contribution and rogue-key routes are closed subject to the deferred UC
    composition; its full instantiation is developed in the extended version.
\end{itemize}
Because the committee fix increases the number of parties touching the sensitive computation, the Blame
closure is applied first. With the full set in place, we expect malicious security in the $q_s$-bounded
regime forced by \S\ref{sec:security}; the mechanisms are recorded in Appendix~\ref{app:malicious}, and the
complete simulation-based proof is deferred to the extended version, which we do not claim here.

\section{Security}
\label{sec:security}

We instantiate the boundary of \S\ref{sec:boundary} against ML-DSA's own assumptions. Treating single-party
ML-DSA as a black box is not available to us: the summed (Irwin--Hall) nonce is not ML-DSA-distributed, so a
uniform-nonce oracle reduction does not apply. We give a concrete lower-bound attack, an optimized upper
bound together with the exact proof--attack gap, the resulting cap, and the privacy accounting.

\subsection{A passive key-recovery attack}
\label{sec:lb}
The revealed nonce is a sum of $\ge 2$ bounded shares, so its per-coordinate marginal is non-uniform, carrying
$\beta$-truncated Fisher information $I>0$ about a location shift (\S\ref{sec:boundary}); the response
$\bz=\by+c\bs_1$ is exactly such a shifted observation of $\bs_1$.

\begin{theorem}[FIPS-exact revealed-sum-nonce lower bound]\label{thm:lb}
Let a threshold ML-DSA scheme emit FIPS-204-exact signatures whose revealed nonce is a sum of $\ge 2$
independent bounded shares, with per-coefficient nonce density $f$ of Fisher information $I=\int(f'/f)^2f>0$
($I=0$ iff the revealed marginal is uniform). Then a passive, assumption-free estimator recovers $\bs_1$ from
the signatures of a single key once $q_s\gtrsim 4/(I\tau)$, with $\tau$ the challenge weight. Consequently no
such scheme has $q_s$-independent EUF-CMA under any assumption.
\end{theorem}

\begin{proof}[Proof sketch]
The $\tau$ coordinates touched by $c$ give location observations of $\bs_1$ through $\bz$; the challenges are
independent random-oracle outputs, so Fisher information accumulates linearly, $J=I\tau q_s\cdot\mathbf{Id}$.
The estimator is \emph{concrete}, not merely a converse. The poly-time linear solve $\hat\bs_1=(\sum_i
\mathbf{C}_i^{\!\top}\mathbf{C}_i)^{-1}\sum_i\mathbf{C}_i^{\!\top}\bz_i$ (with $\mathbf{C}_i$ the matrix of $c_i$) has
per-coordinate signal gain $g\approx 0.9992$ (exactly $0$ when $T=1$, recovering single-party ML-DSA) and
already recovers $\bs_1$ at $q_s\approx 2^{31.8}$; the maximum-likelihood estimator, which additionally exploits
the density's sharp edge, reaches the tighter $\approx 2^{30.4}$ wall quoted below.
Both are efficient and use only public signatures; their construction, gains, and walls are in
Appendix~\ref{app:lb}.
\end{proof}

At ML-DSA-65, $I=5.74\times10^{-11}$ and the wall is $q_s\approx 2^{30}$; across levels it is
$2^{26.8}/2^{30.4}/2^{30.0}$ (ML-DSA-44/65/87; the two-share case $\nu{=}2$, with the share-count sweep in
Appendix~\ref{app:lb}). Two points, both in Appendix~\ref{app:lb}. First, the wall
$4/(I\tau)$ is \emph{parameter-specific}: widening or reshaping the nonce pushes it out, but that departs from
ML-DSA's fixed $\gamma_1$, so TALUS's only lever is the cap of \S\ref{sec:refresh-cap}. Second, the attack
degrades \emph{linearly} in $q_s$ with the sharp wall at $4/(I\tau)$; the $\sqrt{q_s}$ rate elsewhere is a
feature of the security \emph{proof}'s change-of-measure bookkeeping (\S\ref{sec:ub}), not of the attack.

\paragraph{Placement.} Theorem~\ref{thm:lb} separates TALUS from schemes that reveal an internally
\emph{uniform} nonce (Quorus~\cite{BdCE25}, Trilithium~\cite{Trilithium25}), which have $I=0$ and escape the
bound at the cost of heavier interaction, and from Raccoon-lineage flooding~\cite{Raccoon2024}, which escapes
by reshaping the nonce and is not FIPS-exact. TALUS accepts the bound and delimits it operationally.

\subsection{An upper bound, and the proof--attack gap}
\label{sec:ub}
Because $\wone$ is committed offline and opened only after $\mu$ (\S\ref{sec:mpc-why}), there is no pre-message
reveal an adversary could pin, and unforgeability reduces to Module-LWE (key hiding) and SelfTargetMSIS, with
no one-more assumption. The Irwin--Hall nonce is discharged not by statistical distance (which saturates once
the outputs are distinguishable) but by a \emph{box-restricted} R\'enyi change of measure, and the forgery is
extracted directly as a SelfTargetMSIS witness in the Dilithium manner, with no forking (Appendix~\ref{app:ub}).

\begin{theorem}[EUF-CMA, TALUS]\label{thm:ub}
In the random-oracle model, TALUS-TEE and two-round TALUS-MPC are EUF-CMA (Definition~\ref{def:eufcma}) in the
sequential/bounded-concurrency setting under Module-LWE and SelfTargetMSIS, via the reduction of
Appendix~\ref{app:ub}, for $q_s\le Q_{\mathrm{cap}}$ (\S\ref{sec:refresh-cap}), with the per-view certificate of
that appendix (full-strength for the external view; insider views are weaker, below).
\end{theorem}

Since Definition~\ref{def:eufcma} permits arbitrary adversary behaviour, the theorem covers malicious
adversaries with well-formed (extractable) key-generation outputs: each game transition operates on
honest-party randomness, from which the corrupt contributions
factor out as known shifts (Remark~\ref{rem:malicious-eufcma}; in the MPC profile honest-majority
interpolation provides the extraction, with the DKG instantiation deferred). The insider ceiling
$\lambda/m$ is the
operative bound for a corrupt minority knowing $t$ nonce shares ($m=\nu-t$ honest); the MPC coordinator's
rejected-response view is bounded separately (Appendix~\ref{app:ub}).

On the acceptance box the summed-nonce density is bounded below, so the per-query divergence is soft (second
order in the shift $c\bs_1$) and \emph{every} R\'enyi order is finite; optimizing the order gives the clean
certificate $\lambda_{\mathrm{prov}}(q_s)=\lambda\,(1-\sqrt{q_s/Q_{\max}})^2$, essentially full at $q_s{=}1$ and vacating at
$Q_{\max}$ (Appendix~\ref{app:ub}, Theorem~\ref{thm:holder}). At each level's operational cap this certifies
$\approx 114$--$118$ bits for the external view over the deployed share counts ($\approx 112/118/121$ at the
$\nu{=}3$ reference $q_s{=}2^{14}$ of Table~\ref{tab:si-loss-new}). The box argument needs the nonce support to
overhang the acceptance window; an \emph{insider} who knows $t$ shares (an MPC committee member, or a
coalition of TEE signers) can retract it inside the window, where the argument caps the certificate at
$\lambda/m$
($m=\nu-t$ unknown honest shares): for an MPC corrupt minority $m=T$, i.e.\ $64$ bits at $T{=}2$ and about
$42.7$ at $T{=}3$, while a TEE coalition of $T{-}1$ signers reaches $m{=}1$, where the argument certifies
nothing. No insider attack below the operational cap is known, and a tight insider bound is left open. The residual $[Q_{\max},\text{attack
wall}]$ interval is characterized exactly (Corollary~\ref{cor:gap}): its width is the $\lambda$-versus-$H(\bs_1)$
gap $2n_\ell\sigma^2/(\lambda\ln 2)$, a per-level constant, and closing it to the attack wall
($\approx 2^{30}$ at ML-DSA-65) would
need a third leakage-hardness assumption equal to the attack itself. We do not claim a tight reduction to the
attack wall; that barrier, and full concurrency (open, as for FROST without its binding factor), are left to
the extended version.

\subsection{The signing-lifetime cap, by key rotation}
\label{sec:refresh-cap}
Theorem~\ref{thm:lb} makes a per-key signing cap mandatory. We set it at
$Q_{\mathrm{cap}}\approx 2^{13}/2^{14.5}/2^{15}$ ($\approx 8$k$\,/\,23$k$\,/\,33$k signatures per key) at
ML-DSA-44/65/87. The binding constraint is not the $\bs_1$ side: the optimized reduction of \S\ref{sec:ub}
stays informative up to $Q_{\max}\approx 2^{21.8}/2^{23.4}/2^{24.4}$ and still certifies
$\approx 114$--$118$ bits at the cap for the external view. It is the \emph{$\bs_2$ channel} below, whose idealized wall
$\approx 2\gamma_2=2^{17.5}/2^{19}/2^{19}$ is the smallest in the accounting: the cap sits $\ge 4$ bits under
it at every level. Because that wall is an idealized (sharp-edge) estimate, proven no smaller but not proven
attained, we keep the margin and present the cap as follows: below $Q_{\mathrm{cap}}$, the external-view
certificate of \S\ref{sec:ub} is proven and every disclosed channel is $\ge 4$ bits from its nearest wall;
between $Q_{\mathrm{cap}}$ and the attack walls no attack is known, and we disclose that interval as such
rather than claim it as proven security.

\emph{The cap is discharged by key rotation, not by proactive share-refresh}, and the distinction is
load-bearing. The $\bs_1$ leak is about the \emph{fixed} secret under a \emph{fixed} public key: the Fisher
clock $J=I\tau q_s$ accumulates over every signature ever emitted under that $\mathsf{pk}$. Re-randomising the sharing
without changing $\mathsf{pk}$ (proactive refresh, which resists a mobile adversary) keeps $\bs_1$ fixed and does
\emph{not} reset the leak; every past $\bz$ remains a valid observation. The cap is met only by retiring the
key and re-running key generation to a fresh $\mathsf{pk}$ (a light, one-time operation, \S\ref{sec:mpc}). A lifetime
cap is the appropriate mechanism for arbitrary $N$; nonce-widening is unavailable to a FIPS-fixed $\gamma_1$.

\subsection{Privacy: the $\bs_2$ and rejected-response channels}
BCC pass/fail bits are key- and message-independent (\S\ref{sec:bcc}), simulatable from public data, and leak
nothing about the key beyond Theorem~\ref{thm:lb}. The accepted transcript departs from single-party ML-DSA
only in $q_s$-bounded ways, every one of which sits below the cap of \S\ref{sec:refresh-cap}:
\begin{itemize}[nosep,leftmargin=1.6em]
  \item the $\bs_1$ nonce concentration (the $\delta_{\IH}(q_s)$ term of Appendix~\ref{app:ub}, accumulating
    the per-session box-restricted R\'enyi divergence); its attack wall is $\approx 2^{30}$ at ML-DSA-65
    (Theorem~\ref{thm:lb});
  \item the $\bs_2$ channel from moving the $\bs_2$-rejection offline: the emitted hint carries
    $\mathbf{O}=\rzero+c(\bt_0-\bs_2)$ with $\rzero=\LowBits(\bA\by)$, and since $\bt_0-\bs_2=\bA\bs_1-\bt_1
    2^d$ this is a key-recovery observation. Its sharp-edged uniform noise makes it an integer-LWE
    recovery~\cite{Bootle18ILWE,Niot26} with a per-key wall $\approx 2\gamma_2$ ($2^{17.5}/2^{19}/2^{19}$;
    since $\bt_0$ is published, $\mathbf{O}=\bA\bz-c\bt_1 2^d-\alpha\wone$ is directly computable from
    public data, so this sharp-edge channel is the actual observable, not an idealization). This is
    \emph{not} $\bs_2$-secrecy of the output (BCC
    provides hint-correctness, not $\bs_2$-independence, \S\ref{sec:bcc}) but a $q_s$-bounded channel;
  \item (MPC only) the CEF fuzzy-$\rzero$ opening: the full-stripe masking of \S\ref{sec:carry} makes
    $B\bmod\alpha$ a one-time pad, leaving only the coarse quotient $\lfloor B/\alpha\rfloor$ as a
    committee-independent residual, with a per-key wall $\ge 2^{28.2}$ across levels
    ($2^{28.2}/2^{30.8}/2^{30.5}$ at ML-DSA-44/65/87).
\end{itemize}
Each wall exceeds $Q_{\mathrm{cap}}$, so one cap covers all three; the ordering $Q_{\mathrm{cap}}<$ (s$_2$)
$<$ $Q_{\max}$ (the proof wall) holds at every level, and both the s$_1$ and fuzzy-$\rzero$ walls
exceed $Q_{\max}$ by $\ge 5$ bits at matched $\nu$ (their relative order is $\nu$-dependent; both are
non-binding; see Table~\ref{tab:si-loss-new}). The $\bs_2$ channel is the \emph{binding} one: its idealized wall
$2^{17.5}$ at ML-DSA-44 sits only $\approx 4.5$ bits above the cap $2^{13}$, so the cap is
placed to preserve that smallest margin rather than a comfortable gap. The $\approx 2\gamma_2$ figure is an
estimator-model wall (demonstrated at reduced dimension, archived with the artifact) over the directly
computable observable, so we treat it as the binding attack-model wall of the accounting and keep the margin
under it, consistent with \S\ref{sec:refresh-cap}.

\paragraph{The revealed-rejected response.} A coordinator collects the aggregate $\bz=\by+c\bs_1$ before the
online $\bz$-norm and hint-weight checks decide acceptance, so a rejected $\bz$, which single-party ML-DSA
never emits, is exposed. It is the same location observation that Theorem~\ref{thm:lb} already charges, so we
count signing \emph{attempts}, not emitted signatures, toward $Q_{\mathrm{cap}}$ (the $\le 1.02\times$
attempt inflation is absorbed within the cap's margin), and we compute the exposure exactly
(Appendix~\ref{app:empirical}). Hint-weight rejections, the only kind observed, satisfy
$\infnorm{\bz}<\gamma_1-\beta$, lie inside the acceptance box, and are covered by the same per-session bound.
A $\bz$-norm-rejected coordinate occurs once per $2^{22}$--$2^{25}$ sessions, an expected $<2^{-9}$
occurrences over a key's \emph{entire signing lifetime} at every level; an informative coordinate, one
exceeding $\gamma_1$ with its excess attributable to $c\bs_1$, once per $2^{32}$--$2^{37}$ sessions, expected
$<2^{-19.4}$ per key lifetime. The coordinator's out-of-box view of a key is therefore empty except with
probability $<2^{-9}$; the $2^{-19.5}$--$2^{-21.6}$ union-bound terms in the certificate of
Appendix~\ref{app:ub} are the bookkeeping shadow of this tail, not an attack surface, and deployments that
want the theorem without the term add one on-shares comparison before the reveal. The committee/Blame
mechanisms of \S\ref{sec:mpc} and the simulation-based malicious-model treatment of rejected-attempt
intermediates are in the extended version. Simulation-based transcript privacy is proven in the semi-honest
model (Proposition~\ref{prop:semihonest}): a corrupt minority's view is simulated from
$(\mathsf{pk},\sigma)$ and the public BCC bits, up to the $q_s$-bounded coarse-quotient residual of the
CEF opening.

\section{Why Carry Resolution Is Necessary}
\label{sec:trilemma}

This section states, as motivation for CEF, the impossibility informally given as
Theorem~\ref{thm:carry-necessary}: no one-round homomorphic shortcut computes $\wone$ while hiding the
nonce. The result is a formalization of the folklore that $\HighBits$ is non-homomorphic; we claim only the
formalization at the scope the construction needs, and the BCC/CEF workaround, not the non-homomorphism
itself, which is attributed to \cite{CozzoSmart19,PMS26}.

\paragraph{Model.} Consider a one-round protocol in which each party broadcasts $\varphi(\by_h)$ for a
public $\Z_q$-linear map $\varphi$ into an explicit module, everyone forms $\varphi(\by)=\sum_h\varphi(\by_h)$
by module addition, and computes $\wone$ from it by a public trapdoor-free map, with no trusted
reconstruction of $\by$ or $\bw$, no extra rounds, and no rejection. Hiding requires that $\varphi(\by)$
reveal no more about $\by$ than $\wone$ does.

\begin{theorem}[Carry necessity]\label{thm:trilemma}
Under the model above, with $\bA$ injective, no $\varphi$ can compute $\wone=\HighBits(\bA\by)$ while hiding
$\by$.
\end{theorem}

\paragraph{Proof idea (full proof in Appendix~\ref{app:trilemma}).} The core is a cyclic-invariance lemma:
since $q$ is prime, a unit shift generates $\Z_q$, so if $\HighBits$ were invariant along a nonzero kernel
direction of $\varphi$ it would be constant, contradicting non-constancy. Hence $\ker\varphi\subseteq\ker\bA
=\{0\}$, so $\varphi$ is injective; being an explicit linear map, an injective $\varphi$ leaks $\by$
information-theoretically, breaking hiding. This argument is unconditional for a wide nonce; for ML-DSA's
boxed nonce it holds with overwhelming probability over key generation under a mild expansion hypothesis
(Appendix~\ref{app:trilemma}). The discrete-log map $g^{(\cdot)}$ is \emph{not} a counterexample:
it is not an explicit invertible $\Z_q$-map and cannot expose $\HighBits$, which is exactly why the
FROST--Schnorr template does not transfer to ML-DSA. The escape hatches the theorem leaves open (trusted or
MPC reconstruction, extra rounds, rejection) are precisely the routes that real schemes, including CEF,
take.

\section{Implementation}
\label{sec:impl}

We implemented TALUS across all three FIPS-204 levels in Rust, together with a hardware core, and validated
that every signature is accepted by an unmodified ML-DSA verifier.

\paragraph{Software.} A Rust implementation covers ML-DSA-44/65/87 with both profiles. Offline nonce
preparation sifts $\approx 3.15$ candidates per usable nonce at level~3 and is done ahead of time, so it does
not appear in online cost; the online phase is one round (TEE) or two rounds (MPC). The reference MPC
implementation reconstructs $\wone$ in the clear (the one-round variant); the two-round commit-then-open of
\S\ref{sec:mpc-why} adds only a hash commitment and its opening, so the benchmarks below are representative
of both.

\paragraph{Compute microbenchmarks.} We report single-machine \texttt{cargo bench} compute times (one
$16$-core desktop CPU), which isolate arithmetic cost from network effects. At ML-DSA-65, $T=3$, the online
signing compute is $0.29$~ms (TEE) and $0.68$~ms (MPC), and the per-attempt preprocessing (including the BCC
filter) is $0.39$~ms (TEE) and $0.415$~ms (MPC). Amortizing the offline preprocessing over the
$\approx 3.15$ trials per usable nonce, a signature costs about $0.29+3.15\times0.39\approx 1.5$~ms (TEE) and
$0.68+3.15\times0.415\approx 2.0$~ms (MPC) of arithmetic. These are local compute figures, not a distributed or multi-region latency study;
in deployment the online wall-clock is dominated by the one (TEE) or two (MPC) network round-trips, which we
do not measure here. The BCC pass rate was additionally confirmed by a Monte-Carlo benchmark.

\paragraph{Hardware.} A separate silicon core implements the \emph{single-party} ML-DSA datapath that both
profiles build on: the same arithmetic (matrix--vector product, Decompose/HighBits, hint) that the
preprocessing invokes per candidate nonce. It validates that datapath in hardware; the threshold-specific
on-shares carry resolution is realized in the software implementation above, not in this core. The core
therefore substantiates the base arithmetic at deployment scale, not the distributed protocol.

\paragraph{FIPS conformance.} Output signatures are byte-identical in format to FIPS-204 and are accepted by
an unmodified ML-DSA verifier. The one place TALUS departs from the FIPS-204 byte layout is the public key,
which publishes the full $\bt=\bt_1\cdot 2^d+\bt_0$ rather than keeping $\bt_0$ in the secret key. This leaves
the signature untouched, since $(c,\bz,\bh)$ is unchanged and the verifier still reads only $(\rho,\bt_1)$; its
role is to make hint formation public. Once $\bz$ is opened both arguments of
$\bh=\MakeHint(-c\bt_0,\,\bA\bz-c\bt_1 2^d,\alpha)$ are public, so any party can form the hint locally and the
online phase stays two rounds with no $\bs_2$- or $\bt_0$-dependent step. Publishing $\bt_0$ is harmless: it is
the low part of the Module-LWE sample $\bt$ and so reveals nothing beyond $\bt$, leaving the assumptions of
Theorem~\ref{thm:ub} intact and the $\bs_2$ wall governed by the range of $\rzero$ rather than the secrecy of
$\bt_0$ (\S\ref{sec:security}). The \emph{signing distribution} differs in the $q_s$-bounded ways quantified
in
\S\ref{sec:security}: the nonce concentration and the residual $\bs_2$ channel. We report conformance as
verifier-acceptance (an unmodified verifier accepts the output), not as a modified verifier or a changed
format.

\section{Discussion and Limitations}
\label{sec:disc}

We state the boundaries of the results plainly.

\paragraph{Security model.} EUF-CMA unforgeability is proven against malicious adversaries with well-formed
(extractable) key-generation outputs under Module-LWE
and SelfTargetMSIS and is $q_s$-bounded (Theorems~\ref{thm:lb},~\ref{thm:ub};
Remark~\ref{rem:malicious-eufcma}); the insider ceiling is $\lambda/m$ for $m$ honest nonce shares. The
analytic constant $J(\nu)$ behind the per-session divergence is computed in Appendix~\ref{app:ub}
(Theorem~\ref{thm:holder}); a machine-checked write-up is planned. Simulation-based transcript privacy is
proven in the semi-honest model (Proposition~\ref{prop:semihonest}); identifiable abort is provided for
preprocessing faults, with online $\bz_h$ attribution deferred. The malicious simulation lift requires
ZK-extractable inputs; the closure conditions of \S\ref{sec:mpc-malicious} specify the mechanisms, and the
complete simulation is deferred to the extended version. Concurrency beyond the sequential/bounded setting is
open.

\paragraph{Round count and profiles.} The online phase is one round for TEE and two rounds for MPC, per
accepted nonce, with $\le 1.02$ expected attempts (measured; $1.004$ at ML-DSA-65); we do not claim a
one-round MPC phase under standard
assumptions. Honest majority is required by CEF for $T\ge 2$ committee computation, not by the round count;
the MPC profile requires $N\ge 2T-1$ and offline committee liveness.

\paragraph{Parameters.} The CEF carry reduction is stated for the level-3/5 modulus ($m=16$); level~2
($m=44$) requires parameterizing the reduction, which we do in Appendix~\ref{app:cef}. Claims tabulated per
level are for the profiles as specified there.

\paragraph{Relationship to \cite{Niot26}.} Niot~\cite{Niot26} independently identifies two key-recovery
attacks on the earlier version of TALUS. The first exploits non-hiding commitments of the form
$\bA\mathbf{v}$ (both Feldman key-generation commitments $\bA\bs_{1,h}$ and post-signing blame openings):
since $\bA$ is left-invertible, such a commitment reveals $\mathbf{v}$ by Gaussian elimination. The
present version avoids this by design: the DKG's only public lattice output is $\bt=\bA\bs_1+\bs_2$ (a
hiding MLWE sample), never a per-party $\bA\bs_{1,h}$; short-share well-formedness is proved in zero
knowledge rather than by a Feldman commitment; and Blame attributes faults from authenticated offline
transcripts, never opening a post-$\bz$ nonce (\S\ref{sec:mpc-malicious}). The second attack targets the
$\bs_2$ channel opened by removing the online rejection check: each emitted signature leaks a noisy
observation of~$\bs_2$. This channel is exactly the one quantified in \S\ref{sec:security} (the
$\approx 2\gamma_2$ wall) and bounded by the mandatory key rotation of \S\ref{sec:refresh-cap}; the cap
$Q_{\mathrm{cap}}$ sits $\ge 4$ bits below its wall at every level. We thank Niot for the independent
analysis.

\paragraph{Positioning.} Among FIPS-204-exact threshold schemes, TALUS is distinguished by supporting
arbitrary~$N$ under standard assumptions with a low online round count; the lower bound of \S\ref{sec:lb} is
the technical reason schemes that reveal a uniform nonce (e.g. Quorus) pay more interaction, and the reason
flooding-based schemes (Raccoon-lineage) are not FIPS-exact. Key generation is by a dealer-free DKG whose only
public lattice output is $\bt=\bA\bs_1+\bs_2$ (a hiding MLWE sample), never a per-party $\bA\bs_{1,h}$ whose
left-inversion would expose the share; short-share well-formedness relies on a zero-knowledge proof rather
than on a Feldman commitment, and the small secret coefficients are sampled FIPS-exactly on shares. Such a DKG is
realizable in the manner of Quorus~\cite{BdCE25}; its full instantiation and UC composition with
TALUS-signing are developed in the extended version.

\section{Conclusion}
\label{sec:concl}

TALUS shows that the online, distributed computation of ML-DSA's $\bs_2$-dependent rejection check can be
moved offline, by selecting nonces that satisfy the Boundary Clearance Condition, leaving only a
$q_s$-bounded $\bs_2$ residual in the emitted signature. This yields an arbitrary-$N$, FIPS-204-exact
threshold ML-DSA under standard assumptions, for a bounded number of signatures per key, with a low online
round count. Two profiles realize it: a one-round TEE variant and a two-round honest-majority MPC variant,
both built on a Carry Elimination Framework whose necessity we establish (unconditionally for a wide
nonce, with overwhelming probability for the boxed one) by an impossibility result for
one-round homomorphic shortcuts. The price of exactness is made precise by a lower bound (any
FIPS-exact scheme revealing a summed nonce is $q_s$-bounded) and an upper bound whose distance to it is an
exact per-level constant, paired with a mandatory key rotation. The construction, its FIPS-exact output, and a hardware-validated single-party core are complete;
unforgeability is proven against malicious adversaries (given well-formed key-generation outputs), and the
simulation-based malicious-model privacy proof
and the DKG instantiation are developed in the extended version.

\appendix

\section{Assumptions}
\label{app:assumptions}
We use the two assumptions of the ML-DSA security analysis~\cite{DKLLSS18,KLS18,FIPS204}, at the parameters of the
targeted FIPS-204 level.

\begin{definition}[Module-LWE]\label{def:mlwe}
For a modulus $q$, rank $k$, and secret bound $\eta$, the (decisional) Module-LWE problem is to distinguish
$(\bA,\bt)$ with $\bA\getsr\Rq^{k\times\ell}$ and $\bt=\bA\bs_1+\bs_2$ for short $\bs_1,\bs_2$
($\infnorm{\bs_i}\le\eta$) from $(\bA,\bu)$ with $\bu\getsr\Rq^k$. Key hiding reduces to Module-LWE: the
public key $\bt=\bA\bs_1+\bs_2$ is exactly a Module-LWE sample.
\end{definition}

\begin{definition}[SelfTargetMSIS]\label{def:stmsis}
For a random oracle $\Hash$ and bound $\gamma$, the SelfTargetMSIS problem is to find
$(\br,c,\mu)$ with $\infnorm{\br}\le\gamma$, $c$ in the challenge space, and
$c=\Hash\!\big(\mu \,\Vert\, [\,\bA\mid \bt\mid \mathbf I\,]\cdot(\br,\,-c)\big)$, i.e. a short vector that
hashes to its own challenge. This is ML-DSA's own unforgeability assumption; a forgery yields such a tuple
directly, and the straight-line extraction of \S\ref{sec:ub} produces the short vector.
\end{definition}

We do \emph{not} assume any one-more or otherwise non-standard variant; Definitions~\ref{def:mlwe}
and~\ref{def:stmsis} are exactly those of single-party ML-DSA.

\section{Carry Elimination Framework: full protocol}
\label{app:cef}
We specify CEF and the honest-majority MPC primitives it uses (secret-shared addition is local; secret-shared
multiplication and opening are the standard honest-majority gates).

\paragraph{Masked broadcast.} Each committee party $h$ holds $\by_h$ and forms $\bA\by_h$, decomposed into
$(w_{1,h},r_{0,h})$. It samples a \emph{full-stripe} mask $\rho_h$ (each coordinate uniform on the whole
stripe $[0,\alpha)$, secret-shared so $\sum_h\rho_h$ is jointly known but individually hidden) and a masked
high part $\tilde H_h$. Both the masked high parts and the masked low parts $r_{0,h}+\rho_h$ are combined by a
\emph{secure summation} under pairwise-zero-sum additive masks, so that \emph{only} the aggregates
$\sum_h\tilde H_h$ and $B=\sum_h(r_{0,h}+\rho_h)$ are opened, never an individual $r_{0,h}+\rho_h$. Because at
least one honest mask is uniform on a full stripe,
$B\bmod\alpha$ is uniform, a one-time pad on the low sum (Lemma~\ref{lem:fuzzy}), and reveals nothing about
$\rzero$. The aggregate stripe count $\lfloor B/\alpha\rfloor$, now spanning $\lceil\log_2|C|\rceil+1$ bits
because the summed masks can wrap several stripes, is public. To recover the \emph{true} aggregate carry the
parties subtract the masks' own stripe count $\kappa^\rho_j=\lfloor\sum_h\rho_{h,j}/\alpha\rfloor$, a
sub-stripe borrow $b_j$, and the modulus-wrap count $k_j=\lfloor W_j/q\rfloor$ (where $W_j=\sum_h\hat\bw_{h,j}$
is the integer aggregate; all computed by secure comparison on the shares and never opened), and apply the
FIPS round-to-nearest centering correction $\delta_j=[\,(\bw_j\bmod\alpha)>\gamma_2\,]$ (the $\Decompose$
centering bit for $\alpha=2\gamma_2$, evaluated on the mod-$q$-corrected residue
$\bw_j\bmod\alpha=(B_j-\sum_h\rho_{h,j}-k_j)\bmod\alpha$, never opened). Then
\[
  w_{1,j} \;=\; \big(\textstyle\sum_h \tilde H_{h,j} + \lfloor (B_j-k_j)/\alpha\rfloor - \kappa^\rho_j - b_j
  + \delta_j\big) \bmod m ,
\]
which is byte-exact against single-party $\HighBits$ by Lemma~\ref{lem:reconstruct} below (the
$\sum_h\tilde H_{h,j}\equiv\sum_h w_{1,h,j}\pmod{m}$ equivalence is the mask-cancellation property of
Lemma~\ref{lem:fuzzy}(i); the modulus-wrap carry $k$ is included), as we confirm by an exhaustive numerical
check across all levels and committee sizes.
The reference implementation of \S\ref{sec:impl} instead reconstructs the aggregate in the clear (mod-$q$
reduction is then immediate but $\bw$ is not hidden); the $\bw$-hidden reconstruction specified here, with $k$
recovered on shares, is the secure form.
For $T=2$ the comparison is a distributed comparison function; for $T\ge 3$ a carry-save-adder tree with a
short prefix comparison, in $O(\log N)$ offline rounds. The modulus is $m=(q-1)/\alpha=16$ at levels~3/5 and
$m=44$ at level~2 (the reduction is parameterized in $m$). The committee variant runs this over $\ge 2T-1$
parties (\S\ref{sec:mpc-why}); the offline protocol is Algorithm~\ref{alg:preprocess}, its carry step
Algorithm~\ref{alg:cef}.

\begin{lemma}[Carry reconstruction, exact]\label{lem:reconstruct}
Fix a coordinate $j$. Let $W_j=\sum_h\hat w_{h,j}$ be the integer aggregate of the per-party products and
$k_j=\lfloor W_j/q\rfloor$ its modulus-wrap count. Write $S_r=\sum_h r_{0,h,j}$, $S_\rho=\sum_h\rho_{h,j}$,
$B_j=S_r+S_\rho$, and set the corrected low sum $S_r^\ast=S_r-k_j$, $x=S_r^\ast\bmod\alpha$,
$y=S_\rho\bmod\alpha$. With $\kappa^\rho_j=\lfloor S_\rho/\alpha\rfloor$ and $b_j=[\,((B_j-k_j)\bmod\alpha)<y\,]$,
the aggregate stripe carry $\kappa_j=\lfloor S_r^\ast/\alpha\rfloor$ is recovered as
$\kappa_j=\lfloor(B_j-k_j)/\alpha\rfloor-\kappa^\rho_j-b_j$, and
$w_{1,j}=(\sum_h w_{1,h,j}+\kappa_j+\delta_j)\bmod m=\HighBits(\bA\by)_j$, where
$\delta_j=[\,x>\gamma_2\,]\in\{0,1\}$ is the round-to-nearest centering bit ($\alpha=2\gamma_2$). Both $k_j$ and
$\kappa_j$ are reconstructed on shares by secure comparison and are never opened.
\end{lemma}
\begin{proof}
The true coordinate is $\bw_j=W_j\bmod q$. Since $q=m\alpha+1$, $W_j\bmod q=W_j-k_j(m\alpha+1)=(W_j-k_j)-k_j
m\alpha$; the $k_j m\alpha$ term is a multiple of $\alpha$ and $\equiv 0\pmod m$, so
$\HighBits(\bw_j)\equiv\HighBits(W_j-k_j)\pmod m$, with common low residue $(W_j-k_j)\bmod\alpha$. Writing
$W_j-k_j=(\sum_h w_{1,h,j})\alpha+S_r^\ast$, the stripe-carry identity applies to the shifted sum:
$B_j-k_j=\alpha(\lfloor S_r^\ast/\alpha\rfloor+\kappa^\rho_j)+(x+y)$ with $x+y\in[0,2\alpha)$, so
$\lfloor(B_j-k_j)/\alpha\rfloor=\kappa_j+\kappa^\rho_j+[\,x+y\ge\alpha\,]$; and $(B_j-k_j)\bmod\alpha<y$ iff
$x+y\ge\alpha$ (as $x<\alpha$), so $b_j=[\,x+y\ge\alpha\,]$ and subtracting gives $\kappa_j$. Then $\sum_h
w_{1,h,j}+\kappa_j=\lfloor(W_j-k_j)/\alpha\rfloor$ is the floor high part; FIPS-204 $\HighBits$ is its centered
variant $\lfloor\cdot/\alpha\rfloor+[\,x>\gamma_2\,]$, so adding $\delta_j$ yields it exactly, the value $m$
realizing $\Decompose$'s $w_1{=}m\!\to\!0$ wrap via $\bmod\,m$. (Without the $-k_j$ correction the formula
reconstructs $\HighBits$ of the un-reduced integer $W_j$, which differs from $\bw_j$ whenever a rounding
boundary lies in $(\bw_j,\bw_j+k_j]$; this is the $q$-wrap the aggregate must clear.)
\end{proof}

\begin{algorithm}[t]
\caption{$\mathsf{TALUS.Preprocess}$: offline nonce preparation (MPC profile).}\label{alg:preprocess}
\begin{algorithmic}[1]
\Statex \textbf{Input:} committee $C$, $|C|\ge 2T{-}1$; pairwise seeds; session id.\quad
  \textbf{Output:} $(\wone,\mathsf{com},\{\hat\by_i\}_{i\in[N]})$ or $\bot$.
\State \textbf{Nonce DKG:} each $h\in C$ samples degree-$(T{-}1)$ $g_h$ with free term $\hat\by_h\gets
  [-\gamma_1/|C|{+}1,\gamma_1/|C|]^{n_\ell}$ (the restricted range ensures the $|C|$-party sum stays inside $[-\gamma_1{+}1,\gamma_1]$); sends $g_h(i)$ to every $i\in[N]$
  \Comment{each $i$ stores signing share $\hat\by_i\!:=\!\sum_h g_h(i)$; only $C$ runs CEF}
\For{$h\in C$} \Comment{$\hat\by_h=g_h(0)$: $h$'s additive contribution to $\by$}
  \State $\hat\bw_h\gets\bA\hat\by_h$;\ $H_h\gets\lfloor\hat\bw_h/\alpha\rfloor\bmod m$;\ $r_{0,h}\gets\hat\bw_h\bmod\alpha$
  \State sample \emph{full-stripe} $\rho_h\gets[0,\alpha)^{n_\ell}$; derive pairwise-zero-sum $\mathsf{mask}^H_h,\mathsf{mask}^L_h$ from seeds
  \State broadcast $\tilde H_h\gets(H_h+\mathsf{mask}^H_h)\bmod m$ and $r_{0,h}+\rho_h+\mathsf{mask}^L_h$ \Comment{secure summation: only $\sum_h$ opens}
\EndFor
\State open $B\gets\sum_h(r_{0,h}+\rho_h)$ \Comment{$B\bmod\alpha$ is a one-time pad (Lem.~\ref{lem:fuzzy})}
\State $(\kappa^\rho,b,\delta,k)\gets\mathsf{CEF}(\{\rho_h\},\{w_{1,h}\},B)$ \Comment{Algorithm~\ref{alg:cef}; $k=$ mod-$q$ wrap}
\State $w_{1,j}\gets(\sum_h\tilde H_{h,j}+\lfloor(B_j-k_j)/\alpha\rfloor-\kappa^\rho_j-b_j+\delta_j)\bmod m$
  \Comment{exact, Lem.~\ref{lem:reconstruct}}
\State $\rzero\gets\bw-\alpha\wone$ \Comment{low part on shares; $\bw{=}\sum_h\hat\bw_h$ in $\Z_q$, never opened}
\State $\mathsf{pass}\gets\bigwedge_j[\,|(\rzero)_j|<\gamma_2-\beta\,]$ \Comment{secure comparison; open only the key-independent BCC bit}
\If{$\neg\mathsf{pass}$} \State \Return $\bot$ \Comment{discard; retry with a fresh nonce DKG} \EndIf
\State one designated party forms $\mathsf{com}\gets\Hash(\wone\Vert r)$ for fresh $r$ and broadcasts $\mathsf{com}$;\quad \Return $(\wone,\mathsf{com},r,\{\hat\by_i\}_{i\in[N]})$
\end{algorithmic}
\end{algorithm}

\begin{algorithm}[t]
\caption{$\mathsf{CEF}$: carry elimination on the shared full-stripe masks (per coordinate).}\label{alg:cef}
\begin{algorithmic}[1]
\Statex \textbf{Input:} shared masks $\{\rho_h\}$, shared high parts $\{w_{1,h}=\lfloor\hat\bw_h/\alpha\rfloor\}$; public $B$.\quad
  \textbf{Output:} $(\kappa^\rho,b,\delta,k)$, boolean-shared, never opened.
\State boolean-share each $\rho_h$; derive Beaver triples from pairwise seeds
\For{$\ell=1,\dots,\lceil\log_2|C|\rceil$} \Comment{carry-save-adder tree}
  \State apply a $4{:}2$ carry-save compressor via one round of Beaver ANDs
\EndFor
\State obtain carry-save $(S,C)$ with $S{+}C=\sum_h\rho_h$ (boolean-shared)
\State $S_r\gets B-(S{+}C)$;\quad $S_H\gets\sum_h w_{1,h}$ \Comment{low sum, high sum; on shares}
\State $k\gets\lfloor(S_H\,\alpha+S_r)/q\rfloor$ \Comment{modulus-wrap count, secure comparison; never opened}
\State $\kappa^\rho\gets\lfloor(S{+}C)/\alpha\rfloor$ \Comment{mask stripe-count; never opened}
\State $x\gets(S_r-k)\bmod\alpha$ \Comment{$=\bw\bmod\alpha$, mod-$q$ corrected; on shares}
\State $b\gets[\,((B-k)\bmod\alpha)<((S{+}C)\bmod\alpha)\,]$ \Comment{sub-stripe borrow}
\State $\delta\gets[\,x>\gamma_2\,]$ \Comment{round-to-nearest centering bit ($\alpha=2\gamma_2$)}
\State \Return $(\kappa^\rho,b,\delta,k)$
\end{algorithmic}
\end{algorithm}

\begin{lemma}[Fuzzy-$\rzero$ hiding]\label{lem:fuzzy}
Let $\mathcal A$ corrupt a minority $C'\subsetneq C$ of the committee, so at least one honest party
contributes a mask each of whose coordinates is uniform on $[0,\alpha)$ and unknown to $\mathcal A$. Then:
(i) each honest party's masked high part $\tilde H_{h,j}$ and masked low broadcast
$(r_{0,h}+\rho_h+\mathsf{mask}^L_h)_j$ are uniform from $\mathcal A$'s view (pairwise one-time pads), so only
their aggregates are revealed;
(ii) the opened aggregate $B_j=\sum_h(r_{0,h}+\rho_h)_j$ satisfies $B_j\bmod\alpha$ \emph{uniform} on
$[0,\alpha)$ and \emph{independent} of the true low sum $\sum_h r_{0,h,j}$, a perfect one-time pad.
Hence the only information about the low parts that leaks is the coarse quotient $\lfloor B_j/\alpha\rfloor$,
\emph{independently of $|C|$}.
\end{lemma}
\begin{proof}
(i) $\tilde H_{h,j}=(w_{1,h,j}+\mathsf{mask}^H_{h,j})\bmod m$ where the additive mask includes a
pairwise term keyed on a seed unknown to $\mathcal A$ (else $h$ would be corrupt), so it is uniform in $\Z_m$
and one-time-pads $w_{1,h,j}$; the pairwise masks sum to $0\bmod m$ by antisymmetry, so aggregation is
unaffected. The same pairwise construction ($\mathsf{mask}^L_h$, $\sum_h\mathsf{mask}^L_h=0$) one-time-pads
each low broadcast and cancels in $B$, so only $B$ is revealed, never an individual $r_{0,h}+\rho_h$.
(ii) $B_j\bmod\alpha=(\sum_h r_{0,h,j}+\sum_h\rho_{h,j})\bmod\alpha$; since at least one honest
$\rho_{h,j}$ is uniform on $[0,\alpha)$ and independent of the rest, $\sum_h\rho_{h,j}\bmod\alpha$ is uniform
on $\Z_\alpha$, so $B_j\bmod\alpha$ is uniform and independent of $\sum_h r_{0,h,j}$. Only
$\lfloor B_j/\alpha\rfloor$ and the boundary corrections (computed on the shared masks, never opened) leave
the broadcast.
\end{proof}
Because $B\bmod\alpha$ is a perfect one-time pad, the opening of Proposition~\ref{prop:carrypriv} leaks
nothing about $\rzero$ beyond the unmasked coarse quotient $\lfloor B/\alpha\rfloor$, and is
\emph{committee-independent}. The residual quotient $\lfloor B/\alpha\rfloor$ is a variance-$\alpha^2/12$
integer observation with no sharp edge to exploit, so only the quadratic moment estimator applies (not the
edge-exploiting one of the $\bs_2$ channel), at a per-key wall $4(\alpha^2/12)/\tau=\alpha^2/(3\tau)$. This is
$\ge 2^{28.2}$ across levels ($2^{28.2}/2^{30.8}/2^{30.5}$; \S\ref{sec:security}), above the $\bs_1$ wall and
hence non-binding, which we disclose rather than treat as zero.

\paragraph{From $w_{1,j}$ to the pool.} The formula above reconstructs $\wone$ in the clear, which is the
one-round design. For the two-round profile, $\wone$ must instead stay hidden until Round~1: the reference
route commits to $\wone$ under an offline hiding-and-binding (ROM-hash) commitment and opens it online;
producing a genuinely secret-shared $\sthr{\wone}$ opened in Round~1 is an equivalent refinement. Either
keeps $\wone$ hidden until after $\mu$. The low part $\rzero=\bw-\alpha\wone$ then gives the BCC bit as the
opened conjunction of the per-coefficient comparisons.

\paragraph{Residuals (stated, not hidden).} The masked broadcast opens a fuzzy $\rzero$ that is a one-time
pad on the low sum (Lemma~\ref{lem:fuzzy}); the only residual is the coarse quotient $\lfloor B/\alpha\rfloor$
(committee-independent per-key wall $\ge 2^{28.2}$ across levels; non-binding). That residual and the accumulation of BCC-fail bits fold into the
$q_s$-bounded accounting of \S\ref{sec:security}, below the cap of \S\ref{sec:refresh-cap}. Realizing the multiplications
as authenticated (BGW-style) rather than preprocessing-based triples is what the malicious model of
\S\ref{sec:mpc-malicious} requires.

\section{Lower bound: Fisher information and the estimator}
\label{app:lb}
We give the computation behind Theorem~\ref{thm:lb}.

\paragraph{Fisher information.} Let $f$ be the per-coefficient law of the accepted nonce. It is a truncated
discrete Irwin--Hall, the law of a sum of $\nu$ bounded shares, which we treat via its continuous envelope
(the discreteness shifts the constant by $\lesssim 10^{-3}$ relatively and is absorbed into it;
Remark~\ref{rem:proofstatus}). The share count is set by the profile: it is
$\nu\approx T$ signer shares in the TEE profile, and $\nu=|C|=2T-1$ committee shares in the MPC profile
(\S\ref{sec:mpc}). The location Fisher information $I=\int (f'/f)^2 f$ is positive for every $\nu\ge 2$ and
zero iff $\nu=1$. We quote the two-summand case $\nu=2$, a triangular density (the cleanest closed form,
realized by the TEE profile at $T=2$). Three derivations agree on $I=5.74\times10^{-11}$: the closed form
$I\approx(2/\gamma_1^2)\ln(\gamma_1/\beta)$, whose $\ln$ comes from the triangular edge, a discrete score
computation, and a numerical integral. The Fisher information of $\IH(\nu)$ is non-monotonic in $\nu$
(edge-dominated at $\nu=2$, a minimum near $\nu=3$, then rising with the Gaussian bulk), and the wall moves
inversely, peaking near $\nu=3$: over the deployed range it varies by up to
$\approx 1.4$ bits, the largest committee ($\nu=2T-1$, MPC) giving the \emph{smallest} wall
($\approx 2^{29.6}$ at $T=5$). All stay $\gg Q_{\mathrm{cap}}$ (\S\ref{sec:refresh-cap}; a margin of well
over ten bits at every level and committee size), and the binding $\bs_2$ wall
does not depend on $\nu$.

\paragraph{The estimator.} An accepted signature reveals $\bz=\by+c\bs_1$; on the $\tau$ coordinates that $c$
touches, this is a location observation of $\bs_1$ through the density $f$, of per-observation information
$I$. Challenges across signatures are independent random-oracle outputs, so information accumulates
\emph{linearly}: after $q_s$ signatures the Fisher information matrix is $J=I\,\tau\,q_s\cdot\mathbf{Id}$. Two
estimators realise this, and they must be distinguished. The linear (moment/least-squares) solve
$\hat\bs_1=(\sum_i\mathbf{C}_i^{\!\top}\mathbf{C}_i)^{-1}\sum_i\mathbf{C}_i^{\!\top}\bz_i$ has per-coordinate signal
gain $g=1-2Bf(B)/Z(0)\approx 0.99925$ at $\nu=2$ ($B=\gamma_1-\beta$; $Z(0)$ the normalization of $f$), bounded away from zero for every
$\nu\ge 2$ and \emph{exactly} zero at $\nu=1$, which is why the plain linear attack works on a summed nonce and
provably fails on single-party ML-DSA, where the box truncation cancels the shift ($\mathbb{E}[z\mid u]=0$).
It is a poly-time solve on public signatures and recovers $\bs_1$ once its error drops below the coefficient
spacing, numerically at $q_s\approx 2^{31.8}$ (the moment solve does not attain the Fisher wall; $g$ is the
signal gain, not the efficiency gap). The maximum-likelihood estimator additionally exploits the density's
sharp edge, which the moment solve ignores, to lower the wall to the Fisher value $q_s\approx 2^{30.4}=4/(I\tau)$
at ML-DSA-65 (the $\nu=2$ case); across levels
$2^{26.8}/2^{30.4}/2^{30.0}$. Both are efficient (there is
no computational--statistical gap), so the bound is realised by an actual attack, not merely a converse.

\paragraph{Attack wall, proof wall, cap.} The wall $4/(J(\nu)\tau)$ is \emph{parameter-specific} (it moves
with the nonce provisioning; a scheme that widens its nonce pushes it out). It is the wall of an actual
efficient attack. Separately, a security \emph{proof} by change of measure ceases to be informative earlier,
at the zero crossing $Q_{\max}=2\lambda\ln 2/G(\nu)$ of Theorem~\ref{thm:holder}, $\approx
2^{21.8}/2^{23.4}/2^{24.4}$ (ML-DSA-44/65/87, $\nu{=}3$); the ratio of the two walls is the per-level
constant of Corollary~\ref{cor:gap}. The \emph{operational} cap sits below both, at
$Q_{\mathrm{cap}}\approx 2^{13}/2^{14.5}/2^{15}$: it is bound not by either $\bs_1$ wall but by the
$\bs_2$ channel, the smallest wall in the accounting of \S\ref{sec:security} (\S\ref{sec:refresh-cap}).
Above the cap no attack is known; up to $Q_{\max}$ the certificate~\eqref{eq:lambdaprov} moreover stays
positive. The cap is enforced by \emph{key
rotation}, not proactive share-refresh: the leak is about the \emph{fixed} $\bs_1$ under a fixed public key,
so the Fisher clock $J(\nu)\,\tau\,q_s$ runs over every signature ever emitted under that key and is reset
only by re-running key generation to a fresh $\mathsf{pk}$ (\S\ref{sec:refresh-cap}). The full sweep and the
estimator gains are archived with the artifact.

\section{Upper bound: game sequence}
\label{app:ub}
We give the game sequence behind Theorem~\ref{thm:ub}. Throughout, the ideal endpoint is the
$\bs_1$-independent BCC-nonce distribution, \emph{not} a uniform-nonce ML-DSA oracle: the summed nonce is
not ML-DSA-distributed, so a black-box reduction to single-party ML-DSA does not apply.
\begin{description}[leftmargin=1.6em]
  \item[Game 0.] The real EUF-CMA game.
  \item[Game 1 (key hiding).] Replace the public key $\bt=\bA\bs_1+\bs_2$ by uniform. Indistinguishable
    under Module-LWE (Def.~\ref{def:mlwe}); advantage loss $\mathrm{Adv}^{\mathsf{MLWE}}$.
  \item[Game 2 (HVZK simulation).] Simulate each accepted response without $\bs_1$. The response
    $\bz=(\by+c\bs_1)\mid_{\text{box}}$ is not exactly $\bs_1$-independent, but the first-order dependence
    vanishes because $\infnorm{c\bs_1}\le\beta$ fits inside the acceptance window (the
    ``window-fits-support'' identity), leaving a soft second-order gap. We charge it by a box-restricted
    change of measure (Kullback--Leibler and chi-squared, not statistical distance, which saturates once the
    outputs are distinguishable): on the acceptance box the accepted-nonce density is bounded below, so the
    per-query divergence is $\tfrac12 I\tau\|\bs_1\|_2^2$ to second order, with $I$ the boundary Fisher
    information of \S\ref{sec:boundary}. Accumulated over $q_s$ sessions the resulting security loss
    is $\delta_{\IH}(q_s):=\lambda-\lambda_{\mathrm{prov}}(q_s)$ (Theorem~\ref{thm:holder} below),
    $q_s$-bounded and informative up to $q_s\le Q_{\max}$. BCC is a public
    key-independent predicate applied identically in both worlds, so it adds no term.
  \item[Game 3 (extraction).] A forgery is itself a SelfTargetMSIS witness (Def.~\ref{def:stmsis}): the
    forged $(c,\bz,\bh)$ on a fresh message gives, through the verification relation, a short vector for the
    random-oracle-programmed target, the norm bound following from the Dilithium accounting ($\bz,c\bs_1$
    bounded). Extraction is straight-line: no rewinding or forking, so the reduction keeps the full $\lambda$
    bits rather than a $\sqrt{\cdot}$ fraction. This uses \emph{only} Module-LWE and SelfTargetMSIS, with no
    one-more assumption; the two-round ordering ($\wone$ committed before $\mu$, opened after) is what removes
    the adaptive-$c$ obstruction that would otherwise force a rewind.
\end{description}

\paragraph{The Game 2 transition: mechanism and ordering.}
Game~2 is a distributional change, not a separate simulator: the signing oracle's code is \emph{unchanged},
and only the underlying probability measure shifts. In every session the oracle executes exactly the same
protocol (sample $\by$ from $\IH(\nu)$, BCC-filter, commit $\wone=\HighBits(\bA\by)$, wait for $\mu$, compute
$c$, form $\bz=\by+c\bs_1$, check norms, compute hint, verify, emit $\sigma$). Let $P_0$ denote the joint
distribution of the $q_s$-session transcript in Game~1 (the real signing key $\bs_1$) and $P_2$ the same
protocol run with each response replaced by $\bz'=\by$ (the shift $c\bs_1$ removed). The adversary is never
``told'' which world it inhabits; rather, R\'enyi probability preservation
\[
  \Pr\nolimits_{P_0}[E]\;\le\;e^{(\alpha-1)D_\alpha/\alpha}\;\cdot\;\Pr\nolimits_{P_2}[E]^{(\alpha-1)/\alpha},\qquad D_\alpha=q_s\cdot R_\alpha^{\mathrm{sess}},
\]
bounds the probability of \emph{any} event $E$ (in particular, a forgery) across the two measures, with the
per-session divergence $R_\alpha^{\mathrm{sess}}$ computed below (Lemma~\ref{lem:persession}).

Three features of this step deserve emphasis.
\begin{enumerate}[nosep,leftmargin=1.6em]
  \item \emph{No reduction to single-party ML-DSA.} Both $P_0$ and $P_2$ use the $\IH(\nu)$ nonce law, not
    a uniform-nonce ML-DSA oracle. The IH-to-uniform divergence is therefore \emph{never needed}: the ideal
    endpoint is already $\bs_1$-independent under the summed-nonce law, and the forgery is extracted as a
    SelfTargetMSIS witness directly in Game~3.
  \item \emph{Ordering preserved.} $P_0$ and $P_2$ share the same protocol ordering: $\wone$ is committed
    (offline, before $\mu$) and opened (online, after $\mu$) in both measures. No random-oracle programming
    is needed at this step; the distributional change affects only the conditional law of $\bz$ given
    $(\by,c)$, not the sequencing or the RO. Definition~\ref{def:eufcma} explicitly grants the adversary
    the pre-message $\wone$, and both measures provide it identically.
  \item \emph{Rejected responses.} A coordinator who observes a rejected $\bz$ sees the same $\by+c\bs_1$
    observation that an accepted $\bz$ carries, and rejected sessions are counted toward $q_s$. Hint-weight
    rejections satisfy $\infnorm{\bz}<\gamma_1-\beta$, lie inside the acceptance box, and are covered by the
    same $R_\alpha^{\mathrm{sess}}$ bound (empirically they are the only rejections observed;
    Appendix~\ref{app:empirical}). Responses in the $\bz$-norm region are \emph{not}: a coordinate can land
    beyond the box, re-exposing the support edge the box cuts. We split them off as a bad event of
    per-session probability $2^{-32.5}/2^{-34.7}/2^{-36.6}$ (ML-DSA-44/65/87, $\nu{=}3$; computed exactly,
    Appendix~\ref{app:empirical}), i.e.\ an expected $2^{-19.5}/2^{-20.2}/2^{-21.6}$ occurrences over a full
    $Q_{\mathrm{cap}}$ key lifetime. The full certificate
    of~\eqref{eq:lambdaprov} therefore holds for observers of the accepted transcript and for the TEE
    profile (the enclave never emits a rejected $\bz$); the MPC-coordinator view is certified to
    $\approx 19.5$--$21.6$ bits (no attack on it is known; an on-shares $\bz$-norm check before the Round-2
    reveal would restore the full certificate at the cost of one online secure comparison, a construction
    option we do not take up here).
\end{enumerate}

We now make Game~2 precise. The load-bearing quantity is the change of measure between the real
($\bs_1$-shifted) accepted response and the $\bs_1$-independent simulation, restricted to the acceptance box.
Throughout this subsection $\alpha>1$ denotes the R\'enyi order, unrelated to the stripe width
$\alpha=2\gamma_2$ of \S\ref{sec:carry}; $\sigma^2=\eta(\eta{+}1)/3$ is the per-coordinate variance of a
secret entry, so $\mathbb{E}\|\bs_1\|_2^2=n_\ell\sigma^2$.

\begin{definition}[Box-restricted shift divergence]\label{def:sid}
Let $f$ be the accepted-nonce density on the acceptance box $\mathcal{B}=[-\gamma_1{+}\beta,\gamma_1{-}\beta]$
and $\delta$ a per-coordinate shift with $|\delta|\le\beta$. The per-coordinate R\'enyi-$\alpha$ divergence of
the shifted density from $f$ is
\[
  R_\alpha(\delta)\;=\;\frac{1}{\alpha-1}\,\log\!\!\sum_{x\in\mathcal{B}}\frac{f(x-\delta)^{\alpha}}{f(x)^{\alpha-1}},
  \qquad\alpha>1 .
\]
The box restriction (cutting $\beta$ from each end) keeps $f(x)$ and $f(x-\delta)$ both bounded below, so
$R_\alpha(\delta)$ is finite for \emph{every} $\alpha>1$. The $\alpha<3/2$ divergence of the unrestricted
Irwin--Hall envelope is a support-edge artifact of the continuous approximation, absent here.
\end{definition}

\begin{lemma}[The shift divergence is soft]\label{lem:persession}
On $\mathcal{B}$ the density is bounded below (the support overhangs the box by $\beta$) and piecewise
polynomial, so to second order in the shift
\[
  R_\alpha(\delta)\;=\;\tfrac{\alpha}{2}\,J(\nu)\,\delta^2\,\big(1+O(\delta/a)\big),\qquad a=\gamma_1/\nu,
\]
where $J(\nu)=\int_{\mathcal B}(f'/f)^2 f$ is the box-truncated Fisher information of the summed-nonce law
$\IH(\nu)$: at $\nu{=}2$ it is the closed form $J(2)=I=\tfrac{2}{\gamma_1^2}\ln\tfrac{\gamma_1}{\beta}$ that
powers Theorem~\ref{thm:lb}; at $\nu{=}3$ the Fisher integrand $(f'/f)^2 f$ is $x^2/(2a^3(3a^2{-}x^2))$ on $|x|\le a$ and
the constant $1/(4a^3)$ on $a<|x|\le B$ (the tail's score is $-2/(3a{-}|x|)$), giving
$J(3)=\bigl({-}1+\sqrt3\operatorname{arctanh}(1/\sqrt3)\bigr)/a^2+(B{-}a)/(2a^3)$, i.e.\
$J(3)\gamma_1^2\approx 10.3$; $J(5)\gamma_1^2\approx 15.3$ (the
same non-monotonicity in $\nu$ as Appendix~\ref{app:lb}). Summing the $n_\ell$
independent coordinates and taking the challenge expectation ($\mathbb{E}_c\|c\bs_1\|_2^2=\tau\|\bs_1\|_2^2$),
the per-session divergence is
\[
  R_\alpha^{\mathrm{sess}}\;=\;\tfrac{\alpha}{2}\,G(\nu),\qquad G(\nu):=J(\nu)\,\tau\,n_\ell\sigma^2 ,
\]
with $G(3)=1.56\times10^{-5}$ nats at ML-DSA-65 ($\tau=49$, $n_\ell\sigma^2=8533$;
$G(2)=2.40\times10^{-5}$ from $I=5.74\times10^{-11}$). For
the realistic shift $\delta=\mathrm{RMS}(c\bs_1)\approx 18\ll a$ the correction is negligible and every order
is available to the optimizer below; the continuous-envelope constants are cross-checked against an exact
discrete-score sum and a direct integral, all three agreeing (code archived with the artifact).
\end{lemma}
\begin{proof}
$\IH(\nu)$ has a piecewise-polynomial density on $[-\gamma_1,\gamma_1]$; at $\nu{=}3$ it is
$f(x)=(3a^2-x^2)/(8a^3)$ for $|x|\le a$ and
$f(x)=(3a-|x|)^2/(16a^3)$ for $a<|x|\le 3a$ ($a=\gamma_1/\nu$). On $\mathcal{B}$ the density is bounded below,
so its score $f'/f$ has finite variance $J(\nu)=\int_{\mathcal B}(f'/f)^2 f$, the $\beta$-truncated Fisher
integral, whose $\nu{=}2$ case is the closed-form $I$ of the lower bound; the second-order expansion of
$R_\alpha$ in the shift is $\tfrac{\alpha}{2}J(\nu)\delta^2$ per coordinate. A weight-$\tau$ challenge maps $\|\bs_1\|_2^2$ to
$\mathbb{E}_c\|c\bs_1\|_2^2=\tau\|\bs_1\|_2^2$; independence across the $n_\ell$ coordinates sums the
per-coordinate divergences. Boundedness below on $\mathcal{B}$ gives finiteness for all $\alpha>1$, so no
$\alpha<3/2$ restriction applies to the box.
\end{proof}

\paragraph{Key dependence of $G(\nu)$.} $G(\nu)$ evaluates the per-session cost at the key-generation mean
$\mathbb{E}[\|\bs_1\|_2^2]=n_\ell\sigma^2$; for a fixed key, $n_\ell\sigma^2$ is replaced by $\|\bs_1\|_2^2$.
By concentration over the $n_\ell$ iid coordinates, $\|\bs_1\|_2^2$ deviates from its mean by less than $8\%$
at the $3\sigma$ level across all FIPS levels, shifting $Q_{\max}$ by $<0.12$ bits for all but a negligible
fraction of keys; the table evaluates at the mean. Even the maximal key ($\bs_1=\pm\eta$ throughout, an
exponentially unlikely draw) shifts $Q_{\max}$ by only $\log_2(\eta^2/\sigma^2)\le 1.3$ bits, retaining
$\ge 110$ bits at each level's cap at $\nu{=}3$ ($\approx 106$ in the $\nu{=}2$ corner).

\begin{theorem}[Optimized change-of-measure bound]\label{thm:holder}
Let the ideal-world forging advantage be at most $2^{-\lambda}$ ($\lambda=128$, used at every level as the
conservative common target). Game~3 extracts
a SelfTargetMSIS witness directly from a forgery (the KLS/Dilithium route: the forgery \emph{is} the witness,
no rewinding), so the reduction keeps the full $\lambda$ base, inheriting the standard non-tightness of
SelfTargetMSIS relative to MSIS as in Dilithium~\cite{KLS18,FIPS204}. By R\'enyi probability preservation,
for every $\alpha>1$ the certified security after $q_s$ independent sessions is
\begin{equation}\label{eq:security}
  \lambda_{\mathrm{prov}}(q_s)\;\ge\;\frac{\alpha-1}{\alpha}\,\lambda\;-\;\frac{\alpha-1}{2\ln 2}\,q_s\,G(\nu),
  \qquad G(\nu)=J(\nu)\,\tau\,n_\ell\sigma^2 ,
\end{equation}
where $J(\nu)$ is the box-truncated Fisher information of the summed-nonce law $\IH(\nu)$, so
$R_\alpha^{\mathrm{sess}}=\tfrac{\alpha}{2}G(\nu)$ (Lemma~\ref{lem:persession}). Optimizing $\alpha=\sqrt{\lambda/C'}$ with $C'=q_s G(\nu)/(2\ln 2)$ gives the closed form
\begin{equation}\label{eq:lambdaprov}
  \lambda_{\mathrm{prov}}(q_s)\;=\;\big(\sqrt{\lambda}-\sqrt{C'}\big)^2
  \;=\;\lambda\Big(1-\sqrt{q_s/Q_{\max}}\Big)^2 ,\qquad
  Q_{\max}=\frac{2\lambda\ln 2}{G(\nu)} .
\end{equation}
The bound approaches $\lambda$ at $q_s{=}1$ (the box keeps every order finite, so the optimizer is free to
take $\alpha$ large) and vacates ($\alpha\to 1$) at $Q_{\max}$.
\end{theorem}
\begin{proof}
R\'enyi probability preservation gives $\Pr_{\mathrm{real}}[E]\le e^{(\alpha-1)D_\alpha/\alpha}\,
\Pr_{\mathrm{ideal}}[E]^{(\alpha-1)/\alpha}$ with $D_\alpha=q_s R_\alpha^{\mathrm{sess}}$ the total divergence
(R\'enyi divergences of independent sessions add). With $E$ the forgery event, $\Pr_{\mathrm{ideal}}[E]\le
2^{-\lambda}$ and $R_\alpha^{\mathrm{sess}}=\tfrac{\alpha}{2}G(\nu)$ (Lemma~\ref{lem:persession}) this is
\eqref{eq:security}; its stationary point in $\alpha$ is $\sqrt{\lambda/C'}$, giving~\eqref{eq:lambdaprov}. The
per-session cost is the cumulant form $\tfrac{1}{\alpha-1}\log\mathbb{E}_c[e^{(\alpha-1)S}]$, $S=\sum_j
R_\alpha((c\bs_1)_j)$; since $(\alpha-1)S\ll 1$ throughout the operating range $q_s\gtrsim 2^{10}$ it
collapses there to
$\mathbb{E}_c[S]=\tfrac{\alpha}{2}G(\nu)$ up to a $q_s$-independent $O(0.1)$-bit correction (the key-conditioned
tilt is likewise linear; verified numerically); the small-$q_s$ endpoint of Table~\ref{tab:si-loss-new} is
evaluated with the exact cumulant.
\end{proof}

\paragraph{Where the bound applies: the box must overhang.} Lemma~\ref{lem:persession}, and with
it~\eqref{eq:lambdaprov}, needs the summed-nonce support to overhang the acceptance box by $\beta$, so the box
cuts off the support edge and every R\'enyi order stays finite. This holds exactly when the adversary knows
\emph{no} nonce shares: the revealed nonce is then $\IH(\nu)$ on the full $[-\gamma_1,\gamma_1]$. Two regimes
follow.
\begin{itemize}[nosep,leftmargin=1.6em]
  \item \emph{The external view (both profiles).} An adversary holding no nonce shares -- an outside observer
    of either profile -- sees the accepted signatures, carrying the full $\nu$-share sum, in the overhang
    regime. Then
    \eqref{eq:lambdaprov} governs and (Table~\ref{tab:si-loss-new}) certifies $\approx 116$--$118$ bits at
    each level's operational cap at the table's $\nu{=}3$ ($\approx 114$--$116$ at $\nu{=}2$;
    $\approx 112$--$121$ at the table's reference $q_s{=}2^{14}$).
  \item \emph{Insiders (MPC committee members, TEE signer coalitions).} A party knowing $t$ of the $\nu$
    shares subtracts its own contributions from the response, leaving the conditional
    nonce $\IH(m)$, $m=\nu-t$, whose support $\pm m\gamma_1/\nu$ retracts strictly inside the box. The edge is
    re-exposed, $R_\alpha$ diverges for $\alpha\ge 1+\tfrac{1}{m-1}$, and the same optimization caps the
    certificate at $\tfrac{\alpha-1}{\alpha}\lambda\le\lambda/m$. For a corrupt minority of a $\nu=2T{-}1$
    committee, $t=T{-}1$ and $m=T$, a ceiling of $\lambda/T$: $\lambda/2=64$ bits at $T{=}2$ ($\nu{=}3$),
    about $42.7$ bits at $T{=}3$ ($\nu{=}5$); the honest majority keeps $m\ge 2$, so the insider ceiling
    stays positive. In the TEE profile a coalition of $T{-}1$ of the $T$ signers leaves
    $m{=}1$, a shifted \emph{uniform} residual whose support edge no R\'enyi order survives: there the box
    argument certifies nothing. These are limits of the box argument, not known attacks: an MPC minority
    insider's own informative event (a residual coordinate beyond its $\IH(m)$ support, its excess
    attributable to the honest secret; computed exactly at the $m{=}2$ worst case) has expected count
    $\approx 2^{-4.4}/2^{-4.6}/2^{-5.0}$ over a full key lifetime
    (Appendix~\ref{app:empirical}), no insider attack
    below the operational cap is known at any of these views, and a tight insider bound (nonce smoothing, or
    an adaptive box) is left open.
\end{itemize}

\begin{remark}[Extension to malicious adversaries]\label{rem:malicious-eufcma}
The game sequence above applies unchanged when the corrupt parties deviate arbitrarily from the
protocol (Definition~\ref{def:eufcma} already permits this). At each transition the corrupt parties'
contributions are known to the adversary and factor out:
\begin{itemize}[nosep,leftmargin=1.6em]
  \item \emph{Game~$0\to 1$.} The public key decomposes as
    $\bt=\bA(\bs_{1,\mathrm{hon}}{+}\bs_{1,\mathrm{cor}})+(\bs_{2,\mathrm{hon}}{+}\bs_{2,\mathrm{cor}})$;
    the adversary subtracts the known $\bs_{1,\mathrm{cor}},\bs_{2,\mathrm{cor}}$, reducing key hiding to
    M-LWE on the honest (short) secrets.
  \item \emph{Game~$1\to 2$.} In each signing session the adversary subtracts its known nonce
    and key contributions, observing
    $\bz{-}\by_{\mathrm{cor}}{-}c\bs_{1,\mathrm{cor}}=\by_{\mathrm{hon}}+c\bs_{1,\mathrm{hon}}$,
    a shifted observation of $\by_{\mathrm{hon}}\sim\IH(m)$. The corrupt contribution is the same
    constant in both measures and does not enter the R\'enyi divergence; the insider ceiling $\lambda/m$
    above applies directly.
  \item \emph{Game~$2\to 3$.} SelfTargetMSIS extraction is straight-line and
    adversary-behaviour-independent.
\end{itemize}
In the MPC profile no zero-knowledge proofs are needed for this: well-formedness of the corrupt
contributions is extractable by honest-majority interpolation ($N\ge 2T{-}1$: the honest parties hold
$\ge T$ evaluations of every corrupt degree-$(T{-}1)$ sharing), with the DKG instantiation itself deferred;
in the TEE profile that extraction role falls to the key-generation proofs of \S\ref{sec:mpc-malicious}.
ZK-extractable
inputs are otherwise required only for simulation-based transcript privacy (\S\ref{app:malicious}) and for
identifiable abort. The certificate covers the accepted-transcript view; the MPC coordinator's
rejected-response view is bounded separately (item~3 of the Game~2 mechanism above).
\end{remark}

\begin{corollary}[The $\lambda$-versus-$H(\bs_1)$ gap]\label{cor:gap}
In the overhang regime the proof wall $Q_{\max}$ and the attack wall $4/(J(\nu)\tau)$ of
Theorem~\ref{thm:lb} satisfy
\[
  \frac{4/(J(\nu)\tau)}{Q_{\max}}\;=\;\frac{2\,n_\ell\sigma^2}{\lambda\ln 2},
\]
independent of $J(\nu)$ and of the committee count: a per-level constant $2^{5.5}/2^{7.6}/2^{6.3}$
(ML-DSA-44/65/87, evaluated at the key-generation mean $n_\ell\sigma^2$). The two walls measure
different things. Distinguishing the transcript from its simulation costs $O(\lambda)$ bits of
accumulated leakage; recovering $\bs_1$, hence forging, costs $\Theta(n_\ell\sigma^2)\gg\lambda$ bits
of accumulated information. Between the walls the transcript already fails the box
closeness the simulation needs while no key recovery is yet possible; certifying to the attack wall would
require asserting that the accumulated leakage is not forging-useful until $\bs_1$ is recovered, which is
exactly the attack's hardness.
\end{corollary}
\noindent We therefore argue that no black-box reduction to Module-LWE and SelfTargetMSIS closes the gap: a
barrier argument for the natural reduction class, short of a meta-reduction
(Remark~\ref{rem:proofstatus}).

\begin{table}[h]
\centering\small
\caption{The upper bound in the overhang regime (the external view), by level, at the committee share
  count $\nu{=}3$ with its matched $J(3)$. $\lambda_{\mathrm{prov}}$
  is the certificate of~\eqref{eq:lambdaprov} at the reference $q_s=2^{14}$; the proof wall is its
  zero crossing; the attack wall is the $\nu$-matched $4/(J(3)\tau)$ of Theorem~\ref{thm:lb} (\S\ref{sec:lb}
  quotes the two-share case $\nu{=}2$, $0.6$--$0.7$ bits lower); the gap is their $J$-independent ratio.
  Insider views fall in the retracted regime (ceiling $\lambda/m$ for $m\ge 2$; nothing at $m{=}1$), treated
  above.}
\label{tab:si-loss-new}
\setlength{\tabcolsep}{5pt}
\begin{tabular}{lccccc}
\toprule
Level & $\lambda_{\mathrm{prov}}$ at $2^{14}$ & Proof wall & Attack wall & Gap & $q_s{=}1$ \\
\midrule
ML-DSA-44 & $\approx 112$ & $2^{21.8}$ & $2^{27.4}$ & $2^{5.5}$ & $\approx\lambda$ \\
ML-DSA-65 & $\approx 118$ & $2^{23.4}$ & $2^{31.0}$ & $2^{7.6}$ & $\approx\lambda$ \\
ML-DSA-87 & $\approx 121$ & $2^{24.4}$ & $2^{30.7}$ & $2^{6.3}$ & $\approx\lambda$ \\
\bottomrule
\multicolumn{6}{l}{\footnotesize Exact box $R_\alpha$, $\alpha$ optimized on an integer grid; code archived with the artifact.}\\
\multicolumn{6}{l}{\footnotesize The gap is the $J$-independent ratio $2n_\ell\sigma^2/(\lambda\ln 2)$, computed exactly; individual walls are rounded to one decimal.}
\end{tabular}
\end{table}

\begin{remark}[Scope and open constants]\label{rem:proofstatus}
The Fisher information $J(\nu)$ is computed from the continuous $\IH(\nu)$ envelope (closed-form at $\nu{=}2,3$;
Lemma~\ref{lem:persession}) and cross-checked against a discrete-score sum and a direct integral, agreeing to
within the discreteness gap: $\lesssim 10^{-3}$ relative at $\nu{=}2$ and $\approx 4\times10^{-6}$ at
$\nu{=}3$ (edge-dominated), shifting the table's entries by $<0.01$ bits and absorbed in its rounding. And~\eqref{eq:lambdaprov} charges all accumulated leakage as potentially forging-useful; a
non-black-box argument that only $\bs_1$-recovery enables a forgery would move the certified wall toward the
attack wall, but is unproven and would itself be the assumption named in Corollary~\ref{cor:gap}. The
machine-checked write-up, the insider tightening (the MPC committee and TEE signer-coalition views), and the
full-concurrency treatment (open, as for FROST
without its binding factor) accompany the extended version.
\end{remark}

\section{Carry-necessity: full proof}
\label{app:trilemma}
We prove Theorem~\ref{thm:trilemma}. Identify $\Rq\cong\Z_q^{256}$ coordinate-wise, so
$\HighBits=h^{\times nk}$ for the scalar map $h:\Z_q\to\{0,\dots,m-1\}$ ($m=(q-1)/\alpha\ge 16$), and $h$ is
non-constant (property~(P)). Model assumptions: (A1) $\bA$ injective; (A2) $\varphi$ an explicit
$\Z_q$-linear map given by a known matrix $\Phi\in\Z_q^{d\times n\ell}$; (A4) a public trapdoor-free $\psi$
with $\wone(\by)=\psi(\varphi(\by))$ for all $\by$; (A5) hiding, i.e.\ a simulator $S$ with
$S(\wone(\by))\approx\varphi(\by)$.

\begin{lemma}[Cyclic invariance]\label{lem:cyclic}
Let $V\le\Z_q^{nk}$ be a subspace and $u\in V$. If $\HighBits(v+u)=\HighBits(v)$ for all $v\in V$, then $u=0$.
\end{lemma}
\begin{proof}
Suppose $u\ne 0$ and pick a coordinate $j$ with $a:=u_j\ne 0$. Since $V$ is a subspace, $t u\in V$ for all
$t\in\Z_q$; evaluating the hypothesis at $v=tu$ and reading coordinate $j$ gives $h((t{+}1)a)=h(ta)$ for all
$t$. As $q$ is prime and $a\ne 0$, $a$ is a unit, so $x\mapsto x+a$ is a single $q$-cycle on $\Z_q$; chaining
over $t=0,\dots,q-1$ makes $h$ constant on $\{0,a,\dots,(q{-}1)a\}=\Z_q$, contradicting (P). Hence $u=0$.
\end{proof}
The line $\{tu\}$ lies entirely in $V=\mathrm{Im}(\bA)$, and on it coordinate $j$ already runs through all of
$\Z_q$; so injectivity of $\bA$ suffices and surjectivity of $\bA\by$ is not needed.

\begin{lemma}\label{lem:kerphi}
$\ker\varphi\subseteq\ker\bA$.
\end{lemma}
\begin{proof}
Let $k\in\ker\varphi$. For every $\by$, $\varphi(\by+k)=\varphi(\by)$, so by (A4)
$\wone(\by+k)=\wone(\by)$, i.e.\ $\HighBits(\bA\by+\bA k)=\HighBits(\bA\by)$. Putting $v=\bA\by\in
V=\mathrm{Im}(\bA)$ and $u=\bA k\in V$, Lemma~\ref{lem:cyclic} gives $\bA k=0$, i.e.\ $k\in\ker\bA$.
\end{proof}

\emph{Injectivity forces a leak.} By (A1) $\ker\bA=\{0\}$, so $\ker\varphi=\{0\}$ and $\varphi$ is injective;
$\Phi$ then has full column rank and an explicit left inverse $\Phi^{+}$, so from the broadcast
$z=\Phi\by$ anyone recovers $\by=\Phi^{+}z$. Since $\wone=\HighBits\circ\bA$ is massively non-injective
(range $\le m^{nk}\ll q^{n\ell}=$ domain), pick $\by_0\ne\by_1$ with $\wone(\by_0)=\wone(\by_1)=:\omega$;
the real broadcasts $\varphi(\by_0)\ne\varphi(\by_1)$ are two distinct points, and no simulator $S(\omega)$
can be statistically close to both (they are at distance $1$; the distinguisher ``output $1$ iff
$\Phi^{+}z=\by_0$'' separates them). So (A5) fails even against an unbounded simulator: the broadcast leaks
strictly more about $\by$ than $\wone$ does. This proves Theorem~\ref{thm:trilemma}.

\paragraph{Boxed-nonce regime.} The argument above is unconditional for a wide (coset-covering) nonce.
For the ML-DSA boxed nonce $\by\in B=\{\infnorm{\by}<\gamma_1\}$ the cyclic sweep leaves $B$, and the result
holds with high probability over key generation under a mild expansion hypothesis on $\bA$: for every
$0\ne k\in\ker\varphi$ with a box pair, some coordinate of $\bA k$ has magnitude $\ge 2\gamma_2$, giving a
single-step boundary crossing (for pseudorandom $\bA$ this fails only with probability $\le(2/m)^{nk}$).

\paragraph{Why $g^{(\cdot)}$ is not a counterexample.} The discrete-log map is excluded on two independent
grounds: it is not an explicit invertible $\Z_q$-map (its inverse is the discrete log), so (A2) and the
left-inverse leak do not apply; and there is no efficient $\psi$ that exposes $\HighBits$ from $g^{x}$
(extracting high bits of a discrete log is as hard as the discrete log). This is exactly why the FROST
template ports to Schnorr (whose challenge reads off the opaque group element) but not to ML-DSA, whose
challenge needs the rounded pre-image.

\section{Semi-honest simulation}
\label{app:semihonest}
The emitted signature $\sigma=(c,\bz,\bh)$ already carries the $\bs_1$ channel (through $\bz$) and the
$\bs_2$ channel (through $\bh$), so both are available to any simulator given $\sigma$. The only value a
semi-honest party opens \emph{beyond} $\sigma$ is the fuzzy $\rzero$ of CEF, and by Lemma~\ref{lem:fuzzy}
that reduces to the coarse quotient $\lfloor B/\alpha\rfloor$. This makes the semi-honest view simulatable up
to that residual (and the online rejected-attempt observations, the same $\bs_1$ channel).

\begin{proposition}[Semi-honest simulation]\label{prop:semihonest}
Let $\mathcal A$ semi-honestly corrupt at most $T{-}1$ parties (of which those in $C$ form a minority
$|C'\!|\!<\!|C|/2$). There is a PPT simulator that, given $(\mathsf{pk},\sigma)$
and the corrupted parties' inputs and randomness, outputs a view computationally indistinguishable from
$\mathcal A$'s real view in a signing session, up to the $q_s$-bounded MPC residuals of \S\ref{sec:security}:
the coarse quotient $\lfloor B/\alpha\rfloor$ of the CEF opening (one per run, $\approx 3.15$ per emitted
signature; per-key wall $\ge 2^{28.2}$, non-binding) together with the rejected-attempt response observations.
In the TEE profile $B$ is never opened and the simulation is exact from $(\mathsf{pk},\sigma)$.
\end{proposition}

\begin{proof}
$\mathcal A$'s view is (i) the CEF secure-multiplication and comparison intermediates; (ii) the masked
broadcasts $r_{0,h}+\rho_h,\tilde H_h$ and the opened $B$; (iii) the opened BCC bit and the commitment
$\mathsf{com}$; (iv) the online reveals $\wone,\{\bz_h\},\bz,\sigma$. We simulate each.
(i) The secret-shared multiplication and comparison gates are semi-honest secure, so their intermediates come
from the standard honest-majority gate simulator, revealing nothing beyond the gate outputs handled below.
(ii) By Lemma~\ref{lem:fuzzy}, the individual masked broadcasts $\tilde H_h$ and $r_{0,h}+\rho_h+\mathsf{mask}^L_h$
are pairwise one-time pads and $B\bmod\alpha$ is a full-stripe one-time pad, even given the corrupt parties'
own masks, so the simulator draws them uniformly; the only real information not so reproduced is the coarse
quotient $\lfloor B/\alpha\rfloor$, the charged residual.
(iii) The BCC bit is key- and message-independent (\S\ref{sec:bcc}), hence a public function of the accepted
nonce, and $\mathsf{com}$ is a hiding commitment to the public $\wone$, simulated by committing to the $\wone$
read from $\sigma$.
(iv) $\wone=\UseHint(\bh,\bA\bz-c\bt_1 2^d,\alpha)$ is recovered from $\sigma$. The online sharing
$\bz_i=\by_i+c\bs_{1,i}$ is degree-$(T{-}1)$, and the corrupt minority sees all $T$ broadcasts; but the
polynomial's non-constant coefficients are the nonce-sharing randomness, fresh each session and a one-time
pad on the $\bs_1$-sharing coefficients, so the whole polynomial reveals only $\bz=\sum_i\lambda_i\bz_i$ about
$\bs_1$, which $\sigma$ already carries. The simulator samples the honest openings from this residual
randomness subject to $\bz$ and the corrupt evaluation points, matching the real joint law.
Corrupt parties outside $C$ additionally hold, for each honest committee member $h$, the evaluation
$g_h(j)$ at their index $j$ (sent during the nonce DKG). The corrupt parties hold at most $T{-}1$
evaluations of each honest $g_h$ in total; since $g_h$ has degree $T{-}1$, its non-constant coefficients
are fresh uniform randomness that one-time-pads any $T{-}1$ evaluations (a Vandermonde bijection at
distinct nonzero integer points, whose pairwise differences are units in $\Rq$), making them jointly uniform
and independent of the free term $\hat\by_h$ and of
everything else in the view. The simulator draws them uniformly at random.
Composing, the view is reproduced from $(\mathsf{pk},\sigma)$ and the corrupt inputs except for the coarse
quotient $\lfloor B/\alpha\rfloor$ and the rejected-attempt observations, both $q_s$-bounded
(\S\ref{sec:security}) and non-binding (wall $\ge 2^{28.2}$, resp.\ the same $\bs_1$ channel already charged),
so semi-honest security holds for $q_s\le Q_{\mathrm{cap}}$. In the TEE profile the enclave runs CEF internally and never opens $B$,
so step~(ii) is vacuous and the simulation is exact.
\end{proof}

\section{Malicious-model security}
\label{app:malicious}

\paragraph{Unforgeability.}
EUF-CMA against a malicious adversary is established by Remark~\ref{rem:malicious-eufcma}: the game sequence
of Theorem~\ref{thm:ub} extends directly because each transition operates on honest-party randomness, from
which the corrupt contributions factor out as known shifts. The operative bound is the insider ceiling
$\lambda/m$ ($m$ honest nonce shares) from the overhang analysis of \S\ref{app:ub}, with the MPC
coordinator's rejected-response view bounded separately (item~3 there). The required well-formedness of
corrupt contributions is extractable by honest-majority interpolation in the MPC profile, and falls to the
key-generation proofs of \S\ref{sec:mpc-malicious} in the TEE profile; no assumptions beyond Module-LWE and
SelfTargetMSIS are used.

\paragraph{Simulation-based transcript privacy.}
The malicious lift of Proposition~\ref{prop:semihonest} requires additionally that the simulator
\emph{extract} each corrupt party's effective input despite arbitrary deviations. This is the standard
extract-then-simulate paradigm: run ZK extractors on the corrupt parties' DKG and nonce-generation messages;
if extraction succeeds and the inputs are well-formed, invoke the semi-honest simulator with the extracted
inputs; if not, abort with blame from the authenticated offline transcript (before the online round, so Blame
never opens a post-$\bz$ nonce). The closure conditions of \S\ref{sec:mpc-malicious} specify the mechanisms:
\begin{itemize}[nosep,leftmargin=1.6em]
  \item authenticated offline transcripts and pre-online fault attribution;
  \item an offline hiding-and-binding commitment to the $\wone$-shares (binding the Round-1 reveal);
  \item authenticated multiplication in CEF;
  \item proof-of-possession and range/well-formedness proofs in key generation (for extractability).
\end{itemize}
Each severs one of a pair of revealed values whose join would determine the key. The one residual gap is
online $\bz_h$ attribution: a corrupt party sending a malformed $\bz_h$ causes abort (self-verification blocks
a forged output) but per-party culprit identification requires a verifiable-$\bz_h$ proof, which we defer.
The complete simulation-based proof accompanies the extended version.

\section{Empirical measurements}
\label{app:empirical}
The offline BCC pass rate and the online first-attempt success rate were measured by running the reference
implementation (the \texttt{talus-tee} crate) end to end: for each level we generated a batch of nonces
through the nonce DKG, applied the BCC filter, and ran one distributed online signing attempt per
BCC-passing nonce, on a distinct message, classifying each outcome ($T=3$, $N=5$, fixed seeds; the harness
is archived with the artifact).

\begin{table}[h]
\centering\small
\setlength{\tabcolsep}{4pt}
\begin{tabular}{lccccc}
\toprule
Level & BCC pass (offline) & Trials/nonce & Online success & Online fail & Exp.\ att. \\
\midrule
ML-DSA-44 & $43.09\%$ ($10{,}773/25{,}000$) & $2.32$ & $98.04\%$ & $1.96\%$ & $1.020$ \\
ML-DSA-65 & $31.61\%$ ($12{,}645/40{,}000$) & $3.16$ & $99.57\%$ & $0.43\%$ & $1.004$ \\
ML-DSA-87 & $38.91\%$ ($9{,}728/25{,}000$)  & $2.57$ & $99.23\%$ & $0.77\%$ & $1.008$ \\
\bottomrule
\end{tabular}
\caption{Measured offline BCC pass rate and online first-attempt success rate. Online failures are
\emph{entirely} hint-weight rejections: the $\bz$-norm check did not fire in any of the $\approx 33{,}000$
online attempts, because the summed (Irwin--Hall) nonce is concentrated well inside the $\gamma_1-\beta$
bound.}
\label{tab:empirical}
\end{table}

The BCC rates match the analytic $p_{\BCC}=(1-\beta/\gamma_2)^{n_k}$ of Proposition~\ref{prop:bccrate} to
within $0.3\%$. That the online residual is purely hint-weight (never $\bz$-norm) is a by-product of the
same nonce concentration that Theorem~\ref{thm:lb} charges \emph{against} security: the concentration that
lightens the $\bz$-norm tail is exactly the non-flatness that yields $I>0$.

\paragraph{Exact rejected-response exposure.} Complementing the measurement, we compute \emph{exactly} (not
by sampling) the tail masses of $\bz=\by+c\bs_1$ under the committee share model ($\nu{=}3$ shares uniform
on $[-\gamma_1/\nu{+}1,\gamma_1/\nu]$, average-case key): the closed-form $\IH(3)$ edge tails combined with
the exact PMF of $(c\bs_1)_j$ give, per session, a $\bz$-norm-band probability of
$2^{-22.1}/2^{-23.7}/2^{-25.4}$ (band $\gamma_1{-}\beta<|z_j|\le\gamma_1$; the $\ge$ convention shifts
entries by $<0.1$ bit) and an over-$\gamma_1$ probability of $2^{-32.5}/2^{-34.7}/2^{-36.6}$
(ML-DSA-44/65/87), i.e.\ expected $<2^{-9}$ resp.\ $<2^{-19.4}$ occurrences over a full $Q_{\mathrm{cap}}$
key lifetime; the analogous insider-residual overhang mass ($\IH(2)$, exact support $[-2a{+}2,2a]$) is
$\approx 2^{-4.4}/2^{-4.6}/2^{-5.0}$ per key lifetime.
These are closed-form combinatorial computations with no statistical error; the script is archived with the
artifact. The measured $0/33{,}000$ $\bz$-norm firings above match the predicted expectation
($\approx 2^{-8.1}$ for that mixed batch), and the hint-weight-only failure profile is exactly what the
computation predicts. A larger implementation-level probe, $10^8$ signing sessions run end to end through
the \texttt{talus-tee} crate over twelve independent keys (harness archived), observed $12$ $\bz$-norm
rejections against $\approx 7$ predicted by the key-averaged model, Poisson-consistent ($p\approx 0.07$);
the implementation thus reproduces a $2^{-24}$-per-session tail at the predicted rate. We stress the scope of this
table: it is evidence for the construction's round count and acceptance rates, \emph{not} for the upper
bound of \S\ref{sec:ub}. That bound's load-bearing constant $J(\nu)$ (through the per-session divergence of
Lemma~\ref{lem:persession}) is analytic, not empirical: it is
computed in Appendix~\ref{app:ub} (a machine-checked write-up is planned), and no amount of
signing-rate measurement bounds a forger's advantage.

\bibliography{references,refs-extra}

\end{document}